\documentclass[12pt,onecolumn,letter]{IEEEtran}

\usepackage{latexsym}
\usepackage{amssymb}
\usepackage{amsmath}
\usepackage{multirow}
\usepackage{footnote}
\usepackage{color}
\usepackage{bm}
\usepackage[dvips]{graphicx}
\usepackage{multirow}

\usepackage{algorithm}
\usepackage{algorithmic}

\newtheorem{thm}    {Theorem}
\newtheorem{lem}[thm]{Lemma}
\newtheorem{cor}[thm]{Corollary}
\newtheorem{proposition}[thm]{Proposition}

 \newenvironment{proofof}[1]{\vspace*{5mm} \par \noindent
         \quad{\it Proof of #1:\hspace{2mm}}}{\qed
}

\newcommand{\qed}{\hfill \IEEEQED}
\newenvironment{proof}{%
\noindent{\em Proof.\ }}{%
\hspace*{\fill}\qed \\
\vspace{2ex}}

\def\Pr{{\rm Pr}}

\newcommand{\mc}{-\!\!\!\!\circ\!\!\!\!-}

\def\CD{\mathsf{C}}

\def\argmax{\mathop{\rm argmax}}
\def\argmin{\mathop{\rm argmin}}

\newcommand{\bF}{\mathbb{F}}

\makeatletter
\newcommand{\lleq}{\mathrel{\mathpalette\gl@align<}}
\newcommand{\ggeq}{\mathrel{\mathpalette\gl@align>}}
\newcommand{\gl@align}[2]{
\vbox{\baselineskip\z@skip\lineskip\z@
\ialign{$\m@th#1\hfil##\hfil$\crcr#2\crcr{}_{{}_{(=)}}\crcr}}}
\makeatother

\def\sgn{\mathop{\rm sgn}\nolimits}

\def\Label#1{\label{#1}\ [\ \text{#1}\ ]\ }
\def\Label{\label}

\begin{document}
\title{Quantum-inspired secure wireless communication protocol 
under spatial and local Gaussian noise assumptions}
\author{Masahito Hayashi~\IEEEmembership{Fellow,~IEEE}
\thanks{The material in this paper was presented in part at the 2017 IEEE International Symposium on Information Theory (ISIT 2017),   Aachen (Germany), 25-30 June 2017.}
\thanks{Masahito Hayashi   is with the Graduate School of Mathematics, Nagoya University,
Furocho, Chikusaku, Nagoya, 464-860, Japan,
Shenzhen Institute for Quantum Science and Engineering, Southern University of Science and Technology,
Shenzhen, 518055, China
and
Centre for Quantum Technologies, National University of Singapore, 3 Science Drive 2, Singapore 117542.
(e-mail: masahito@math.nagoya-u.ac.jp)}
}

\markboth{M. Hayashi: Secure wireless communication under spatial and local Gaussian noise assumptions}{}

\maketitle

\begin{abstract}
Inspired from quantum key distribution, 
we consider wireless communication between Alice and Bob 
when the intermediate space between Alice and Bob is controlled by Eve. 
That is, our model divides the channel noise into two parts,
the noise generated during the transmission and the noise generated in the detector.
Eve is allowed to control the former, but is not allowed to do the latter.
While the latter is assumed to be a Gaussian random variable,
the former is not assumed to be a Gaussian random variable.
In this situation, using backward reconciliation and the random sampling, 
we propose a protocol to generate secure keys between Alice and Bob under the assumption that Eve's detector has a Gaussian noise and Eve is out of Alice's neighborhood.
In our protocol, the security criteria are quantitatively guaranteed even with finite block-length code
based on the evaluation of error of the estimation of channel.

\end{abstract}

\begin{IEEEkeywords}
secret key generation,
reverse reconciliation,
post selection, 
noise injection,
wireless communication
\end{IEEEkeywords}

\section{Introduction}
Recently, secure wireless communication attracts much attention as a practical method to realize physical layer security \cite{LH,LPS,LP,YPS,BB11,BBRM,special,BHT,OH,SC,Trappe,Zeng,WX}.
In particular, wire-tap channel model \cite{Wyner,CK79,Csiszar,hay-wire} is considered as a typical model for physical layer security.
In the wire-tap channel model,
the authorized sender, Alice is willing to transmit her message to the authorized receiver, Bob without any information leakage to the adversary, Eve.
In this case, we usually assume that 
the noise in the channel to Eve is larger than that in the channel to Bob.
However, it is not easy to guarantee this assumption under the real wireless communication.
In cryptography, it is usual to consider that
the adversary, Eve is more powerful than the authorized users, Alice and Bob in some sense like RSA cryptography \cite{RSA}.
However, the above wire-tap channel requires the opposite assumption.
So, it does not necessarily have sufficient powers of conviction to assume the above wire-tap channel in real wireless communication.

Instead of wire-tap channel model, we often employ secure key agreement, in which, Alice and Bob generate
the agreed secure key from their own correlated random variables \cite{Mau93,AC93}.
This problem has a similar problem when they generate secure keys via one-way communication from Alice to Bob, 
because they need to assume that the mutual information between 
Alice and Bob is larger than that between Alice and Eve.
Further, although there exist proposals to generate secure key from wireless communication \cite{LLP,YMRSTM,WS,PJCG,CJ},
they do not give a quantitative security evaluation for the final keys under a reasonable assumption advantageous to Eve in a finite-length setting.

On the other hand, many people are studying quantum key distribution (QKD)\cite{BB84}, which enables us to 
generate secure key without any assumption for Eve's performance
when Alice and Bob are allowed to use public channel and 
they do not detect the existence of Eve.
At least, even when Eve has much powerful performance than Alice and Bob, e.g., 
the intermediate space between Alice and Bob might be controlled by Eve,
Alice and Bob can generate secure keys.
Its security evaluation has been done even with finite-length including second order analysis
\cite{SP,Mayers,PRA06,TLGR,NJP12,NJP14}. 
In QKD, after the quantum communication from Alice to Bob, 
Alice and Bob check whether the secure keys can be distilled from random variables generated by 
the initial quantum communication via public channel.
Also, Fung \cite{Fung} et al proposed to use error verification to guarantee the reliability of the final keys
for QKD.
In addition, as pointed in several papers \cite{Lo,BBL,PRA07},
when Eve cannot access the noise in Bob's detector,
information reconciliation with backward can improve the asymptotic key generation rate.
In this way, in the context of QKD, they proposed several advanced methods to generate secure keys 
after the initial transmission. 
However, QKD requires more expensive devices even for Alice and Bob.
Hence, it is not so easy to implement QKD.
Therefore, it is required to propose an alternative secure-key generation protocol of quantum key distribution under a reasonable assumption by using cheaper devices.

In this paper, inspired from these advanced techniques in QKD, 
we propose a protocol to generate quantitatively secure keys between Alice and Bob under a reasonable assumption advantageous to Eve
when Alice and Bob do not detect the existence of Eve.
In this analysis, Eve is allowed to control the intermediate space between Alice and Bob,
however, she is not allowed to access the noise in Bob's detector in a similar way to the analysis in  \cite{Lo,BBL,PRA07}.
The biggest difference from the above analysis in QKD is the assumption that
Eve's detector has non-negligible noise, which cannot be accessed by Eve.
This assumption is stronger than that in QKD, 
but enables us to generate secure key without use of quantum communication.
Since quantum key distribution assumes public channel, Alice and Bob are allowed to use public channel in 
the first step of this paper.
However, Eve might override the signals to Bob or the public channel for spoofing \cite{SC}.
We explain a method to avoid such attack, which requires shared secret randomness with small size.
Further, for efficient realization of the protocol, we additionally impose the following requirements.

\begin{description}
\item[(R1)]
The security of final keys is guaranteed quantitatively
based on acceptable criterion even for cryptography community 
(e.g. the variational distance criterion \cite{R-K} or the modified mutual information criterion)
even though Eve takes the optimal strategy under the above assumption.
Additionally, the formula to derive the security evaluation has sufficiently small calculation complexity.

\item[(R2)]
The calculation complexity of the whole protocol (Protocol 1 given in Section \ref{s31}) is sufficiently small.
\end{description}

This paper is organized as follows.
Firstly, 
we rigorously explain our purpose and our assumption in Section \ref{s01}.
Then, we compare our formulation with existing jamming attacks in Section \ref{s02}.
As the next step, before proceeding to our protocol, we discuss the mathematical structure of our model in Section \ref{s0}.
In Section \ref{s3}, we give our concrete secure protocol
by assuming the public channel.
In particular, the end of Subsection \ref{NHS}, we briefly explain the solution for spoofing.
Section \ref{s33} analyzes the security of the given protocol,
and numerically evaluates the security in a typical case.
Section \ref{s7} is devoted to two kinds of extensions, multiple antenna attack and complex number case.
Section \ref{s21} consider the relation of a model with interference to the additional noise of the eavesdropper channel so that we clarify how our model contain such a interference channel.
Section \ref{s5} gives proofs of statements given above.

\section{Purpose and assumptions}\Label{s01}
Recall that the aim of this paper 
is to propose a protocol to generate quantitatively secure keys between Alice and Bob under a reasonable assumption advantageous to Eve.
Here, our aim is not to always generate secure keys, but is to detect the existence of eavesdropping with high probability when it exists.
That is, when they consider there is no eavesdropper,
their keys are required to be matched and secret.
In other word, it is required to discard their keys when an eavesdropper exists.
Here, the case without eavesdropper means the case when 
the operation of the eavesdropper cannot be distinguished from the natural phenomena. 
So, the natural case, i.e., the case with the natural phenomena,
is very important in our analysis.

In the real setting, it is difficult to identify where Eve attacks the communication between Alice and Bob except for Alice's neighborhood and Bob's detector.
To guarantee the security of the final keys in such a setting,
it is natural to assume the following conditions when Alice sends the $i$-th signal $A_i$.

\begin{description}
\item[(A1)]
The intermediate space between Alice and Bob might be controlled by Eve
while Eve's operation is restricted to satisfy the following conditions.
That is, the information $Y_i$ can be injected by Eve as Fig. \ref{F3} so that
Eve knows the noise $Y_i+ e_B$ added to Bob's detection during transmission in the intermediate space.
Here, we choose the variable $Y_i$ such that its average is $0$.
Hence, the average of the noise is $e_B$.
$e_B$ is independent of $i$ to due to Assumption (A4).

\item[(A2)]
When Bob and Eve detect the $i$-th receiving signal $B_i$ and $E_i$,
independent Gaussian noises $b_B X_{1,i}$ and $b_E X_{2,i}$ are added, respectively,
where $X_{1,i}$ and $X_{2,i}$ are subject to the stranded Gaussian distribution
and independent of $X_{1,i'}$ and $X_{2,i'}$ for $i\neq i'$.
Alice and Bob know the lower bounds of the powers $b_B$ and $b_E$ of their noise.
This assumption is called the {\it local Gaussian noise assumption}.
Since no detector has no detection noise, this assumption is reasonable.
Nobody can control these noises.

\item[(A3)]
Alice and Bob know the lower bound of 
the attenuation $a_E$ for Alice's signal in Eve's detection.
When Eve is out of Alice's neighborhood, this condition holds.
This assumption is called the {\it spatial assumption} for Eve.

\item[(A4)]
The wireless communication between Alice and Bob is {\it quasi static}.
That is, the channel between Alice and Bob is almost constant 
during a specific time interval so called the coherent time \cite[Section 5.4.1]{TV}.
In other words, during the coherent time, the noise can be considered to be independently and identically distributed and to be independent of other variables. 
Also, the attenuation $a_B$ for Alice's signal in Bob's detection.
and the attenuation $a_E$ for Alice's signal in Eve's detection
can be considered to be constants.
It means that Eve does not has ability to change the added nose dependently of the signal 
transmitted by Alice\footnote{This assumption means that the noise added by Eve cannot be adaptively controlled. 
This assumption is natural because such an adaptive noise operation requires much advanced technology.
In fact, in the early stage of studies of QKD, 
they assume that the errors is subject to an identical and independent distribution.
The attack under this condition is called the collective attack in the QKD \cite{BBBGM}.
Hence, it is natural to assume this kind of assumption 
at the first paper of our setting.}.
Therefore, 
the variable $Y_i$ is independent of $A_i$ and  the distribution of $Y_i$ does not depend on $i$.
Also, the average $e_B$ is independent of $i$ and $A_i$.
We also assume that we can send one block of our protocol during the coherent time\footnote{In various protocols, 
a set of pulses or bits treated as one block is called a coding block.
The number of such pulses or bits is called a block length. 
For example, in RSA cryptography, since the arithmetic is based on the public composite $m$, $\log m$ is a  block length.
Our protocol is composed of an error correcting code like an LDPC code.
Since the block length of an LDPC code is from 10000 to 100000,
the block length of our protocol is from 10000 to 100000 when we employ a LDPC code.}.
In the natural case, the noise in the channel from Alice to Bob is a Gaussian noise.

\item[(A5)]
Noiseless public channel between Bob and Alice is assumed.
In practice, it can be realized by a combination of error correcting code and 
noisy wireless channel between Bob and  Alice.
This assumption can be confirmed by authentication as explained in Subsection \ref{NHS}. 
\end{description}

In summary, when we have $m$ transmissions from Alice to Bob, 
there is the following relation among 
Alice's $i$-th sending real variable $A_i$,
Bob's $i$-th receiving real variable $B_i$, and
Eve's $i$-th receiving real variable $E_i$.
\begin{align}
B_i&:= {a_B} A_i+ Y_i+ {b_B} X_{1,i}+ e_B, \Label{4-24-1X} \\
E_i&:= {a_E} A_i+ {b_E} X_{2,i}+ e_E. \Label{4-24-2X}
\end{align}
Here, the coefficients $a_B,b_B$, $e_B$, $a_E$, $b_E$, and $e_E$ are constants 
with physical meaning as Table \ref{table1}. 
To discuss the situation advantageous to Eve,
we assume that Eve's detection has no noise except for the noise inside of her detector as Eq. \eqref{4-24-2X}.
Even though we put $e_E$ to be $0$, there is no information loss.
So, we consider only the case when $e_E$ is $0$ for simplicity.
Also, due to Assumption (A1), Eve knows the value of $Y$ as well as $E$.
Then, additionally we assume the following assumption.
\begin{description}
\item[(A6)]
Eve knows 
all the channel parameters, which are given in \eqref{4-24-1X} and \eqref{4-24-2X}
as well as the value of $Y_i$
\end{description}


\begin{table*}[htb]
\begin{center}
\caption{Summary of parameters.
$c_{AB}$ is covariance between $A$ and $B$.
$v_B$ is variance of $B$.
$v_Y$ is variance of $Y$.
These parameters are known to be Eve.}
\begin{tabular}{|l||l|l|l|l|} \hline
\multirow{2}{*}{Coefficient}   & \multirow{2}{*}{Meaning}  & 
Long time & Treatment in &\multirow{2}{*}{Estimation method}\\
& & period behavior & this paper & \\ \hline \hline
\multirow{2}{*}{$a_B$} & \multirow{2}{*}{Attenuation} & \multirow{2}{*}{Stochastic} & To be estimated & \multirow{2}{*}{$c_{AB}$}\\
&&& by sampling &\\ \hline
\multirow{3}{*}{$a_E$} & \multirow{3}{*}{Attenuation} & \multirow{3}{*}{Stochastic} & Constant (upper & Distance between\\
&&& bound among & Alice and Eve \\
&&& possible values) & (Ass.(A3)) \\ \hline
\multirow{2}{*}{$v_Y$} & Noise amplitude & \multirow{2}{*}{Stochastic} &
To be estimated & \multirow{2}{*}{$v_B-c_{AB}^2-b_B^2$} \\
& during transmission & & by sampling &\\ \hline
\multirow{2}{*}{$b_B$} & Bob's detector  &
\multirow{2}{*}{Constant} & \multirow{2}{*}{Constant} & Performance of Bob's \\
& noise amplitude & & & detector (Ass.(A2)) \\ \hline
\multirow{2}{*}{$b_E$} & Eve's detector   & \multirow{2}{*}{Constant} & \multirow{2}{*}{Constant} & Performance of Eve's\\
& noise amplitude & && detector (Ass.(A2)) \\
\hline
\end{tabular}
\Label{table1}
\end{center}
\end{table*}

\begin{figure}[htbp]
\begin{center}
\scalebox{0.7}{\includegraphics[scale=0.6]{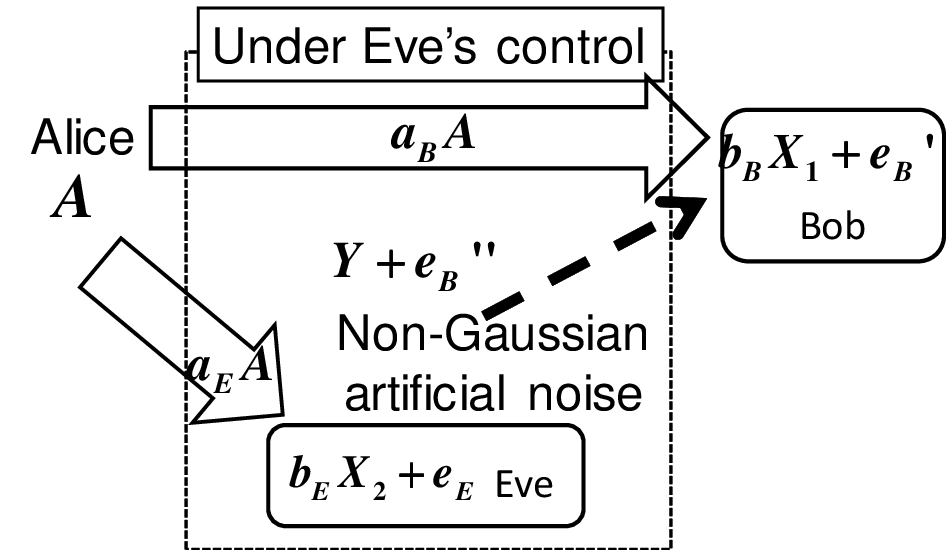}}
\end{center}
\caption{Eve injects artificial noise to Bob's observation.
$e_B=e_B'+e_B''$.}
\Label{F3}
\end{figure}%

Since the information $Y_i$ can be injected by Eve as Fig. \ref{F3},
the attack under Assumptions (A1)-(A6) is called {\it noise injecting attack}.
When Eve is closer to Alice than Bob and 
the performance of Eve's detector is the same as that of Bob's,
the signal-noise ratio of Eve is not smaller than that of Bob so that
secure communication by one-way wire-tap channel is impossible.
We discuss this problem in Section \ref{MDO}.

To overcome this problem, this paper considers two-way protocol like QKD.
In the two-way protocol, the performance depends on the direction of 
information reconciliation (error correction).
When we employ the forward reconciliation,
the performance is the same as the case with one-way wire-tap channel,
as explained in Section \ref{MDO},
then, it cannot realize secure communication in the above case.
Our protocol employs the reverse reconciliation 
after the above information transmission, as will be given in Section \ref{s3}.

Here, we discuss the meaning of coefficients $a_B$ and $a_E$ more deeply.
The coefficients $a_B$ and $a_E$ express the attenuation.
The intensities $a_E^2$ and $a_B^2$ behave as $C d^{-\alpha}$ with positive constants $C$ and $\alpha$
when the distance from Alice's transmitting antenna is $d$,
and have stochastic behavior as fading in long time span\cite[Section 5.4.1]{TV}. 
For example, the free space with no obstacle has the constant $\alpha=2$\cite{TV}.
Due to {\it spatial assumption} (Assumption (A3)), 
Eve's detector is sufficiently far from Alice's transmitting antenna.
So, the relation $d\ \ge d_0$ holds with a certain constant $d_0$
Under this assumption, we can guarantee that $a_E^2 \le C d_0^{-\alpha}$.

On the other hand, the coefficients $b_B$ and $b_E$ can be lower bounded by 
the performance of their detectors due to Assumption (A2).
As explain in Section \ref{s31}, our protocol contains random sampling.
Hence, the coefficient $a_B $ can be estimated as covariance $c_{AB}$ between $A$ and $B$
in the random sampling,
which provides a better estimate than the method based spatial relation between Alice and Bob.
Thus, the meaning of these parameters can be summarized in Table \ref{table1},
while the parameter $v_Y$ will be introduced in Section \ref{s0}.
Some of these parameters are estimated from 
covariance $c_{AB}$ between $A$ and $B$,
variance $v_B$ of $B$, and variance $v_Y$ of $Y$.


\section{Comparison with previous papers}
\subsection{Comparison with other attacks}\Label{s0BA}
Most of existing studies for secure wireless communication were done in the context of 
wire-tap channel.
First, Wyner proposed the model of wire-tap channel \cite{Wyner}.
Then, Csisz\'{a}r and K\''{o}rner extended the model to 
the broadcast channel with confidential messages (BCC) \cite{CK79}, in which
the source node also has a common message
for both receivers in addition to the confidential message for only one receiver. 
Also, Leung-Yan-Cheong and Hellman applied the wire-tap channel to the Gaussian channels \cite{LH}.
The recent paper \cite[Appendix D-C]{HM} showed the strong security in this model.
Liang et al extended these analyses to fading channel \cite{LPS}. 
Then, many preceding papers \cite{LPS,LP,YPS,BB11,BBRM,special,BHT,OH,SC,Trappe,Zeng,WX} studied physical layer security by using wire-tap model.
That is, most of studies for physical layer security in the community of information theory fall in this framework.
However, if we adopt wire-tap model, in order to realize secure communication, we need to assume that 
the mutual information between Alice and Bob is greater than 
the mutual information between Alice and Eve.
However, it is usual that Eve is closer to Alice than Bob.
Hence, it is unnatural to assume such a assumption.

In order to remove this strong assumption, we employ a two-way protocol,
in which, after receiving the signal from Alice, Bob sends a modified information to Alice.
We make detailed comparison between our method and 
the one-way case like the wire-tap channel in Section \ref{MDO}.

In the one-way case, 
it is sufficient to discuss the relation between the above two types of mutual information.
That is, even when Eve makes various types of attacks, 
the security analysis are reduced to 
the analysis to the two channels, the channel from Alice to Bob and the channel from Alice to Eve.
This characterization holds even when both channels have memory
because asymptotic capacity formula was shown only with the form of these two channels
for a general  sequence of channels \cite{H2015}.
However, when we employ a two-way protocol,
we need to be careful to the correlation between the added noise in the channels from Alice to Bob's
and Eve's signals.
Therefore, we compare our model with other attacks in the next subsection.

\subsection{Comparison with other attacks}\Label{s02}
Here, we compare our model with the elementary jamming attack \cite{GLY}.
In the elementary jamming attack, Eve inserts her artificial noise to Bob's detection. However, she does not know the value of the added noise. 
The purpose of jamming attack is to interrupt the communication between Alice and Bob, and is not to eavesdrop the secret information between Alice and Bob.
Hence, in the jamming attack, Eve makes the artificial noise so large that 
the error correction by Alice and Bob does not work.
That is, the jamming attack might make their final keys mismatched.
Since our protocol contains the the process of estimation of channel parameters, 
when Eve makes the elementary jamming attack,
Alice and Bob can detect such a large noise, i.e., the existence of the elementary jamming attack. 

In our noise injecting attack model, Eve is allowed to know the value of the added noise. 
Hence, even when the artificial noise is as small as the natural case,
she might obtain a part of information of the final keys.
Hence, it is difficult for Alice and Bob to find the existence of the noise injecting attack.
That is, Eve in our noise injecting attack model is more powerful than Eve in jamming attack.  
In such a scenario, Alice and Bob need to prepare their protocol so that their information transmission is secure against the most powerful Eve within the scope of their assumption.

Now, we compare our model with channel-hopping jamming attack.
In channel-hopping jamming attack, Eve overrides the signal from Alice to Bob for spoofing \cite{Zeng,SC,AS}.
However, such an attack can be prevented by the authentication between Alice and Bob.
That is, our protocol is secure even against channel-hopping jamming attack
by equipping authentication.

Next, we explain how our assumption covers the case when Eve can change her strategy dynamically.
Alice and Bob can choose the detailed parameter of the secure key distillation protocol (Protocol 2 given in Section \ref{s31}) 
depending on the coding block because our secure key distillation protocol will be done as post processing.
Assumption (A4) means that Eve cannot change her strategy during the coherent time interval.
That is, she can change her strategy only in the next coherent time interval.
Since we estimate channel parameters for each coherent time interval, 
our protocol properly reflects such a dynamical change.
In this way, our assumption is more general than existing attacks and covers various types of attacks.

\section{Mathematical structure}\Label{s0}
\subsection{General case}
Before proceeding to our protocol, 
we discuss the mathematical structure of our model
when Alice independently generates her random variables $A_i$ subject to the standard Gaussian distribution. 
That is, we discuss how to simplify Eve's knowledge for $B$ in this setting.
Then, we discuss how to estimate the distribution of the variable describing Eve's knowledge for $B$.
For this aim, since the variables $A_i$, $Y_i$, $X_{1,i}$, and $X_{2,i}$ are independent,
we consider the single-system description as follows.
\begin{align}
B&:= {a_B} A+ Y+ {b_B} X_{1}+ e_B, \Label{4-24-1G} \\
E&:= {a_E} A+ {b_E} X_{2}. \Label{4-24-2G}
\end{align}
In this section, we assume the models \eqref{4-24-1G} and \eqref{4-24-2G}, in which
the random variables $A$, $X_1$, and $X_2$ are independent standard Gaussian random variables
and $Y$ is an independent variable with average $0$.
In this subsection, we introduce variables $E',E'',A^c,U$.
When we employ the multiple-system description like \eqref{4-24-1X} and \eqref{4-24-2X},
these variables with $i$-th transmission are written
to be $E_i',E_i'',A^c_i,U_i$.
The variables are summarized in Table \ref{TF3}.

When we discuss the cumulative distribution function of a real valued variable $X$, 
we denote it by $F_X$ and treat it. 
In fact, under the condition $Y=y$, we denote the cumulant distribution function of $X$ by
$F_{X|Y=y}$.
In the following discussion, $\Phi_v$ expresses the cumulant distribution function of Gaussian variable with variance $v$.
Firstly, we prepare the following theorem.

\begin{thm}\Label{T2X}
Using the variable
\begin{align}
E':= \frac{a_B a_E}{a_E^2+b_E^2}E+ Y, \Label{Edash}
\end{align}
we have the following relation
\begin{align}
P_{B|(E,Y)=(e,y)}=P_{B|E'=e'},
\end{align}
where 
$e'=\frac{a_B a_E}{a_E^2+b_E^2}e+ y$,
and $P_{X|E'=e'}$ is the conditional distribution for $X$ when $E'$ is $e'$.
Also, $P_{B|E'=e'}$ is the Gaussian distribution
with average $e'+e_B$
and variance $
v_{B|E'}:=\frac{a_B^2 b_E^2}{a_E^2+b_E^2}+b_B^2 $.
\end{thm}

This theorem implies that the noise injecting attack can be reduced to 
the attack only with the random variable $E'$.

\begin{proof}
We introduce the random variable 
\begin{align}
U:= b_E A- a_E X_2, \Label{variU}
\end{align}
which is a Gaussian random variable independent of $Y$ and $X_1$.
Also, since the covariance between $U$ and $E$ is zero,
$U$ is independent of $E$.
The variance of $U$ is $b_E^2+ a_E^2$. 
Since $A= \frac{a_E}{a_E^2+b_E^2}E+ \frac{b_E}{a_E^2+b_E^2}U$,
we have
\begin{align}
B= {a_B} A+ Y + {b_B} X_1+e_B
=\frac{a_B a_E}{a_E^2+b_E^2}E+ Y + \frac{a_B b_E}{a_E^2+b_E^2}U
+ {b_B} X_1+e_B.
\Label{CPG}
\end{align}
Since 
$\frac{a_B b_E}{a_E^2+b_E^2}U+ {b_B} X_1$ is a Gaussian random variable with
variance $\frac{a_B^2 b_E^2}{a_E^2+b_E^2}+b_B^2$
and is independent of $E$ and $Y$, we have
\begin{align}
F_{B|(E,Y)=(e,y)}(\frac{a_B a_E}{a_E^2+b_E^2}e+ y +t+ e_B)
=
\Phi_{\frac{a_B^2 b_E^2}{a_E^2+b_E^2}+b_B^2}(t)
=F_{B|E'=e'}(e' +t+ e_B),
\end{align}
which implies the desired statement.
\end{proof}

To evaluate the amount of the information leaked to Eve, 
we need to estimate the distribution $P_{E'}$ of the random variable $E'$
from the variables $A$ and $B$ accessible to Alice and Bob.
For this aim, we introduce the random variable 
\begin{align}
A^{c}:= Y+{b_B} X_1= B-a_B A -e_B,\Label{Ac}
\end{align}
which is the total noise in $B$ and can be directly estimated by Alice and Bob in the protocol given in 
Section \ref{s3}.
That is, we discuss how to derive the distribution $P_{E'}$ from 
the distribution $P_{A^c}$ of $A^c$.
For this discussion, we employ the 
convolution $F_1 * F_2$ for given two distribution functions $F_1$ and $F_2$, 
which is defined $F_1 * F_2$ as $F_1 * F_2(x):= \frac{d F_1}{dx}(x-y) F_2(y) dy $
when $F_1$ is differentiable.
When $F_2$ is differentiable, the convolution $F_1 * F_2$ is defined as $F_2 * F_1$.
In particular, we utilize Gaussian convolution ${\cal G}_{v}$ defined as
\begin{align}
{\cal G}_{v}[F]:=  \Phi_{v}*F. \Label{IG} 
\end{align}
Since $Y$ is independent of other variables,
the distributions of $P_{A^c}$ and $P_{E'}$ of $A^c$ and $E'$ are given as 
\begin{align}
F_{A^c}={\cal G}_{b_B^2}[F_Y],\quad
F_{E'}={\cal G}_{\frac{a_B^2 a_E^2}{a_E^2+b_E^2}}[F_Y].
\end{align}
When $ \frac{a_B^2 a_E^2}{a_E^2+b_E^2} \ge b_B^2 $,
we can estimate the distribution $P_{E'}$ by applying 
the Gaussian convolution to the distribution $P_{A^c}$ as 
\begin{align}
F_{E'}={\cal G}_{\frac{a_B^2 a_E^2}{a_E^2+b_E^2 }- b_B^2}[F_{A^c}].
\end{align}

However, when 
\begin{align} 
\frac{a_B^2 a_E^2}{a_E^2+b_E^2} < b_B^2 ,
\Label{OT1}
\end{align} 
we cannot apply this method.
Indeed, we can estimate the distribution $P_{E'}$ by applying the Gaussian deconvolution,
which is the inverse operation of the Gaussian convolution \eqref{IG}.
But, it is quite difficult to estimate the amount of the error of our estimate of 
the distribution $P_{E'}$ when we employ the Gaussian deconvolution.
Since our evaluation of the amount of leaked information requires the evaluation of the amount of error in the estimation of the distribution,
we employ the distribution $P_{A^c}$ instead of the distribution $P_{E'}$ as follows.
In this case, we introduce two independent standard Gaussian random variables $Z_1$ and $Z_2$ 
instead of $X_1$ such that $ \sqrt{b_B^2- \frac{a_B^2 a_E^2}{a_E^2+b_E^2}}Z_1
+\frac{a_B a_E}{\sqrt{a_E^2+b_E^2}}Z_2= b_B X_1$.
Hence, instead of \eqref{CPG},
we have 
\begin{align}
B =\frac{a_B a_E}{a_E^2+b_E^2}E+ Y 
+ \sqrt{b_B^2- \frac{a_B^2 a_E^2}{a_E^2+b_E^2}}Z_1
+ \frac{a_B b_E}{a_E^2+b_E^2}U
+\frac{a_B a_E}{\sqrt{a_E^2+b_E^2}}Z_2+e_B.\Label{HUI}
\end{align} 
Since $\frac{a_B a_E}{a_E^2+b_E^2}E+\sqrt{b_B^2- \frac{a_B^2 a_E^2}{a_E^2+b_E^2}}Z_1$
is a Gaussian variable with variance $b_B^2$ and is independent of $Y$, 
the variable
\begin{align}
E'':=\frac{a_B a_E}{a_E^2+b_E^2}E+ Y + \sqrt{b_B^2- \frac{a_B^2 a_E^2}{a_E^2+b_E^2}}Z_1
\Label{Eddash}
\end{align}
satisfies $F_{E''}={\cal G}_{b_B^2}[P_Y]=F_{A^c}$.
Fortunately, 
it is sufficient to know the cumulative distribution function $F_{E''}$ 
in this case instead of $F_{E'}$
because 
$E''$ is more informative with respect to $B$ than $E'$
due to the Markovian chain 
\begin{align}E' \mc E'' \mc B,\Label{E16D}
\end{align}
which follows from \eqref{HUI}.
In fact, when Eve knows the random variable $Z_1$ as well as the random variables $E$ and $Y$,
$E''$ can be considered as Eve's knowledge with respect to $B$ due to the following corollary of
Theorem \ref{T2X}.

\begin{cor}\Label{C-T2X}
Assume \eqref{OT1}.
(Otherwise, we cannot define the variable $E''$.)
Then, we have the Markovian chain $E' \mc E'' \mc B$ and
\begin{align}
P_{B|(E,Y,Z_1)=(e,y,z_1)}=P_{B|E''=e''},\Label{11-3-1C}
\end{align} 
where 
$e'':=\frac{a_B a_E}{a_E^2+b_E^2}e+ y + \sqrt{b_B^2- \frac{a_B^2 a_E^2}{a_E^2+b_E^2}}z_1$.
Also, $P_{B|E''=e''}$ is the Gaussian distribution
with average $e''+e_B$ and variance $a_B^2 $.
\end{cor}

\begin{proof}
We use the same notations as Theorem \ref{T2X}.
Since 
$\frac{a_B b_E}{a_E^2+b_E^2}U
+\frac{a_B a_E}{\sqrt{a_E^2+b_E^2}}Z_2$ is a Gaussian random variable with
the variance $a_B^2$
and is independent of $E$, $Y$, and $Z_1$, 
\eqref{HUI} implies that
\begin{align}
F_{B|(E,Y,Z_1)=(e,y,z_1)}(\frac{a_B a_E}{a_E^2+b_E^2}e+ y+ \sqrt{b_B^2- \frac{a_B^2 a_E^2}{a_E^2+b_E^2}}z_1+ t+ e_B)
=\Phi_{a_B^2 }(t)
=F_{B|E''=e''}(e'' +t+ e_B),
\end{align}
which implies all the desired statements.
\end{proof}



\begin{table*}[htb]
\begin{center}
\caption{Summary of random variables.}
\begin{tabular}{|l||l|l|l|l|} \hline
\multirow{2}{*}{Variable}   & \multirow{2}{*}{Meaning}  & 
Gaussian/ & \multirow{2}{*}{Equation No.} & \multirow{2}{*}{Variance}\\
& & Non-Gaussian && \\ \hline \hline
\multirow{2}{*}{$A$} & Alice's sending & \multirow{2}{*}{Gaussian} 
& \multirow{2}{*}{\eqref{4-24-1G}, \eqref{4-24-2G}} 
& \multirow{2}{*}{$1$} \\
& variable && &\\ \hline
\multirow{2}{*}{$B$} & Bob's receiving & \multirow{2}{*}{Non-Gaussian} 
& \multirow{2}{*}{\eqref{4-24-1G}} & \multirow{2}{*}{$a_B^2+v_Y+b_B^2 $} 
\\
& variable && &\\ \hline
\multirow{2}{*}{$E$} & Eve's receiving & \multirow{2}{*}{Gaussian} 
& \multirow{2}{*}{\eqref{4-24-2G}} &\multirow{2}{*}{$a_E^2+b_E^2 $} \\
& variable && &\\ \hline
\multirow{2}{*}{$Y$} & Eve's receiving & \multirow{2}{*}{Non-Gaussian} 
& \multirow{2}{*}{\eqref{4-24-1G}} &\multirow{2}{*}{$v_Y $} \\
& variable && &\\ \hline
\multirow{2}{*}{$b_B X_1$} & Bob's detector & \multirow{2}{*}{Gaussian} 
& \multirow{2}{*}{\eqref{4-24-1G}} &\multirow{2}{*}{$b_B^2$} \\
& noise && &\\ \hline
\multirow{2}{*}{$b_E X_2$} & Eve's detector & \multirow{2}{*}{Gaussian} 
& \multirow{2}{*}{\eqref{4-24-2G}} &\multirow{2}{*}{$b_E^2$} \\
& noise && &\\ \hline
\multirow{2}{*}{$E'$} & Eve's information & \multirow{2}{*}{Non-Gaussian} 
& \multirow{2}{*}{\eqref{Edash}} &\multirow{2}{*}{$\frac{a_B^2 a_E^2}{a_E^2+ b_E^2}+v_Y$} \\
& in $B$ && &\\ \hline
\multirow{2}{*}{$E''$} & Eve's information & \multirow{2}{*}{Non-Gaussian} 
& \multirow{2}{*}{\eqref{Eddash}} &\multirow{2}{*}{$b_B^2+v_Y$} \\
& in $B $ in case \eqref{OT1} && &\\ \hline
\multirow{2}{*}{$A^c$} & Total noise & \multirow{2}{*}{Non-Gaussian} 
& \multirow{2}{*}{\eqref{Ac}} &\multirow{2}{*}{$v_Y+b_B^2$} \\
& in $B$ && &\\ \hline
\multirow{2}{*}{$U$} & Gaussian variable & \multirow{2}{*}{Gaussian} 
& \multirow{2}{*}{\eqref{variU}} &\multirow{2}{*}{$b_E^2+ a_E^2$} \\
& independent of $E$ && &\\ \hline
\end{tabular}
\Label{TF3}
\end{center}
\end{table*}

\subsection{Gaussian case}
Now, as a typical case, we assume that $Y$ is also a Gaussian random variable with variance $v_Y$
while we do not assume this assumption except for this subsection.
In this case, the possibility of secure key generation can be discussed by comparison of 
the correlation coefficient $\rho_A$ between $B$ and $A$ and the correlation coefficient $\rho_E$ between $B$ and $E$. 
That is, when $\rho_E^2< \rho_A^2$, we can distill secure keys from $A$ and $B$ with backward reconciliation.
By using variance $v_Y$ of $Y$,
the former correlation coefficient $\rho_A$ is calculated as
\begin{align}
\rho_A^2= 
\frac{a_B^2}{a_B^2+v_Y+b_B^2} .\Label{11-3-1}
\end{align}
Since the variance of $E'$ is $\frac{a_B^2 a_E^2}{a_E^2+b_E^2}+b_B^2 $
and the covariance between $E'$ and $B$ is $\frac{a_B^2 a_E^2}{a_E^2+b_E^2}+b_B^2  $,
by using covariance $c_{AB}$ between $A$ and $B$
and variance $v_B$ of $B$, 
the correlation coefficient $\rho_{E'}$ between $E'$ and $B$
can be calculated as
\begin{align}
\rho_{E'}^2
= & \frac{1}{a_B^2+v_Y+b_B^2} 
\Big(\frac{a_B^2 a_E^2}{a_E^2+b_E^2}+ v_Y \Big) \nonumber \\
= & \frac{1}{v_B}
\Big(\frac{c_{AB}^2 a_E^2}{a_E^2+b_E^2}+ v_B -c_{AB}^2-b_B^2 \Big) \nonumber \\
= & 
1-  \frac{c_{AB}^2}{v_B}
+\frac{1}{v_B}\Big(\frac{c_{AB}^2 a_E^2}{a_E^2+b_E^2}-b_B^2 \Big) .\Label{11-3-2}
\end{align}

When we find that the distribution $P_{E'}$ is sufficiently close to the Gaussian distribution,
the security can be approximately evaluated by the above formula.
That is, the comparison between \eqref{11-3-1} and \eqref{11-3-2} clarifies whether secure keys can be distilled.
Now, we have the following lemma.
\begin{lem}\Label{L1}
The inequality
$\rho_{E'}^2 < \rho_A^2$
holds if and only if
\begin{align}
\frac{a_B^2}{v_Y}> \frac{a_E^2}{b_E^2}+1 
\quad \Big(
\hbox{ i.e., }
v_Y< \frac{a_B^2 b_E^2}{a_E^2+b_E^2}
\Big)
\Label{4-24-3}.
\end{align}
\end{lem}

\begin{proof}
The condition
$\frac{1}{a_B^2+v_Y+b_B^2} (\frac{a_B^2 a_E^2}{a_E^2+b_E^2}+ v_Y)
< \frac{a_B^2}{a_B^2+v_Y+b_B^2} $
is equivalent to 
$0<
a_B^2 -(\frac{a_B^2 a_E^2}{a_E^2+b_E^2}+ v_Y)
=
\frac{a_B^2 b_E^2}{a_E^2+b_E^2}- v_Y
=a_B^2(\frac{b_E^2}{a_E^2+b_E^2}- \frac{v_Y}{a_B^2})$.
This condition is equivalent to
$1+\frac{a_E^2}{b_E^2}
=\frac{a_E^2+b_E^2}{b_E^2}< \frac{a_B^2}{v_Y}$.
So, we obtain Lemma \ref{L1}.
\end{proof}

Instead of \eqref{11-3-2}, we calculate 
the correlation coefficient $\rho_{E''}$ between $E''$ and $B$
can be calculated as follows.
Since the variance of $E''$ is $v_Y+b_B^2 $
and the covariance between $E''$ and $B$ is $v_Y+b_B^2 $,
the correlation coefficient $\rho_{E''}$ between $E''$ and $B$
can be calculated as
\begin{align}
\rho_{E''}^2
= & \frac{1}{a_B^2+v_Y+b_B^2} 
(v_Y +b_B^2) \nonumber \\
= & \frac{1}{v_B}(v_B-c_{AB}^2)\nonumber  \\
= & 
1-  \frac{c_{AB}^2}{v_B}.\Label{11-3-3}
\end{align}

\section{Protocol}\Label{s3}
\subsection{Description of protocol}\Label{s31}
\begin{Protocol}                  
\caption{Whole protocol}         
\label{protocol1}      
\begin{algorithmic}
\LECTURE  
Before the protocol, 
Alice and Bob need to know the values of $b_B$, $b_E$, and $a_E$
based on Assumptions (A2) and (A3).
In the following steps, Alice and Bob are allowed to use the noiseless public channel.
(The common choices in the following steps can be done as follows.
For example, Alice randomly decides the choices of the sample data, and 
she sends the information with respect to these choices to Bob via the noiseless public channel.
Also, they exchange common chosen samples data via the noiseless public channel.)

\STEPONE[Initial key transmission]  
Alice generates her $n+2l$ variables $A_1, \ldots, A_{n+2l}$
according to standard Gaussian distribution independently, 
and she sends them to Bob by using the given channel $n+2l$ times.

\STEPTWO[Estimation 1]
After initial communication, Alice and Bob randomly choose common $l$ samples data
$(\bar{A}_1,\bar{B}_1), \ldots, (\bar{A}_l,\bar{B}_l)$
from $({A}_1,{B}_1), \ldots, ({A}_{n+2l},{B}_{n+2l})$.
They obtain the estimates $\hat{e}_B$, $\hat{v}_B$, and $\hat{c}_{AB}$ 
of the average of $B$, the variance of $B$ and the covariance of $A$ and $B$
by using the average, the unbiased variance, and the unbiased covariance
of the common $l$ samples data, respectively. 

\STEPTHREE[Estimation 2]
Alice and Bob randomly choose another $l$ common samples data
$(\tilde{A}_1,\tilde{B}_1), \ldots, (\tilde{A}_l,\tilde{B}_l)$ from the remaining $n+l$ data.
Based on $\hat{e}_B$, $\hat{v}_B$, and $\hat{c}_{AB}$, they obtain the estimates $\hat{P}_{A^c}$ and $\hat{P}_{E'}$ of 
the distributions $P_{A^c}$ and $P_{E'}$ 
when $ \frac{\hat{c}_{AB}^2 a_E^2}{a_E^2+b_E^2} \ge b_B^2 $.
Otherwise, they obtain only the estimate $\hat{P}_{A^c}$ of the distribution $P_{A^c}$, which works as the estimate $\hat{P}_{E''}$ of $P_{E''}$ as well.
Here, they redefine the random variables
$E':=E'+e_B-\hat{e}_B$ and $E'':=E''+e_B-\hat{e}_B$.

\STEPFOUR[Secure key distillation] 
Based on the above estimates,
Alice and Bob apply the backward secure key distillation protocol for $n$ data, which will be explained as Protocol \ref{protocol2}. 
For Protocol \ref{protocol2}, they redefine 
$({A}_1,{B}_1), \ldots, ({A}_n,{B}_n)$
as the remaining $n$ pairs of Alice's and Bob's variables.

\end{algorithmic}
\end{Protocol}

The whole protocol for noise injecting attack is given as Protocol \ref{protocol1}, which runs Protocol \ref{protocol2} (backward secure key distillation protocol) as a subprotocol. 
Before Protocol \ref{protocol2}, we discuss secure key distillation protocols.
Although there exist several methods to asymptotically attain the optimal one way key distillation rate from Gaussian random variables by using suitable discretization \cite{WO,NN,CB,LCV},
there is no protocol to distill secure keys from Gaussian random variables satisfying the following conditions\footnote{To satisfy these two requirements, our analysis is needed to be simple.
For this aim, the random variable $A$ is needed to be subject to the Gaussian distribution.
This property of $A$ allows us to simply describe Eve's information to be $E'$
as shown in Theorem \ref{T2X}.}.
\begin{description}
\item[(B1)] The whole calculation complexity is not so large.
\item[(B2)] A security evaluation of the final key is available with finite block-length.
\end{description}

Since the difficulty of its efficient construction is caused by the continuity,
we employ very simple discretization in our protocol.
Before describing the secure key distillation protocol, we prepare notations for hash functions.
We consider a randomized function $f_H$ from $\bF_2^{n_1}$ to $\bF_2^{n_2}$, where $H$ is the random variable identifying the function $f_H$,
and $m_1:=n_1-n_2$ and $n_2$ are called the {\it sacrifice bit length} and the output length, respectively.
Alice and Bob need to prepare random seeds $H$ to identity the function $f_H$.
The seeds $H$ is allowed to be leaked to Eve. 
A randomized function $f_H$ is called a universal2 hash function
when the collision probability satisfies the inequality 
\begin{align}
\Pr \{ f_H(c)=f_H(c')\} \le 2^{n_2-n_1}
\end{align}
for any distinct elements $c\neq c' \in \bF_2^{n_1}$ \cite{Carter,WC81}.
In the above equation, $\Pr$ expresses the probability with respect to the choice of $H$.
Under these preparations, we give our protocol satisfying the above conditions 
(B1) and (B2) as Protocol \ref{protocol2}.
Notice that the choice of the sacrifice bit length $m_1$ given in \eqref{28-1Y} (\eqref{28-2Y}) does not assume that $Y$ is Gaussian.

\begin{Protocol}                  
\caption{Backward secure key distillation protocol for $n$ data}         
\label{protocol2}      
\begin{algorithmic}
\LECTURE  
Before stating the following steps, Alice and Bob prepare the estimates 
$\hat{P}_{E'}$ (or $\hat{P}_{E''}$), $\hat{P}_{A^c}$, $\hat{e}_B$, and $\hat{c}_{AB}$.
In the following steps, all communications between Alice and Bob are done via the noiseless public channel.
\STEPONE[Discretization]  
Bob converts his random variable ${B}_i-\hat{e}_B$ to $1$ or $0$ by taking its sign, 
i.e.,  he obtains 
the new bit random variable ${B}_i'$ in $\bF_2$ as $(-1)^{{B}_i'}=\sgn ({B}_i-\hat{e}_B)$.
${B'}^n$ is defined as $(B_1', \ldots,B_n')$.
\STEPTWO[Information reconciliation]
Based on the capacity $I[\hat{P}_{A^c},\hat{c}_{AB}]$
of the channel $W_{A|B'}$:
\begin{align*}
W_{A|0}(a):=&
\frac{\frac{1}{\sqrt{2\pi}}e^{-\frac{a^2}{2}} \int_{0}^{\infty}\hat{P}_{A^c}(b- \hat{c}_{AB} a) d b}
{\int_{-\infty}^{\infty}\frac{1}{\sqrt{2\pi}}e^{-\frac{{a'}^2}{2}} 
\int_{0}^{\infty}\hat{P}_{A^c}(b'-\hat{c}_{AB} a') 
d b' d a'}
\\
W_{A|1}(a):=&
\frac{\frac{1}{\sqrt{2\pi}} e^{-\frac{a^2}{2}} \int_{-\infty}^{0}\hat{P}_{A^c}(b-\hat{c}_{AB} a) d b}
{\int_{-\infty}^{\infty}\frac{1}{\sqrt{2\pi}} e^{-\frac{{a'}^2}{2}} \int_{-\infty}^{0}\hat{P}_{A^c}(b'-\hat{c}_{AB} a') d b' d a'},
\end{align*}
Alice and Bob prepare an error correcting code $ \CD \subset \bF_2^n$, 
{\bf whose choice will be explained in Subsection \ref{L89}.}
Bob computes the syndrome as an element $[{B'}^n]$ of the coset space $\bF_2^n/\CD $
from his bit sequence ${B'}^n$, 
calculate its representative element $\alpha([{B'}^n])$ in $\bF_2^n$,
and sends $\alpha([{B'}^n])$ to Alice.
Bob calculates ${B'}^n-\alpha([{B'}^n]) \in \CD$.
Alice applies the error correction to the data $((-1)^{\alpha([{B'}^n])_i}A_{i})_{i=1}^{n}$
so that she obtains the estimate of ${B'}^n-\alpha([{B'}^n]) \in \CD$.
Here, the error correction is based on the channel $\{W_{Z|0},W_{Z|1}\}$.

\STEPTHREE[Privacy amplification] 
Based on $\hat{P}_{E'} $ or $\hat{P}_{E''} $ in addition to 
$\hat{e}_B$ and $\hat{c}_{AB}$,
Alice and Bob decide the sacrifice bit length $m_1$, {\bf whose choice will be given  in \eqref{28-1Y} or \eqref{28-2Y}}.
Then, they apply universal2 hash function to their bits in $\CD$
with sacrifice bit length $m_1$.
They obtain the keys $K$ with length $\dim \CD-m_1 $.
Here, Alice (or Bob) generates the random seeds locally and can send it to Bob (or Alice) via public channel. 

\STEPFOUR[Error verification]
Alice and Bob choose the bit length $m_2$ for error verification, {\bf whose choice will be discussed in Subsection \ref{NHS}}. 
They apply another universal2 hash function to the keys with output length $m_2$.
They exchange their output of the universal2 hash function.
If they are the same,
discarding their final $m_2$ bits from their keys, they obtain their final keys.
If they are different, they discard their keys.
\end{algorithmic}
\end{Protocol}

Finally, we discuss the effects of the stochastic behaviors of the coefficients $a_B$ and $a_E$ due to fading.
Even though this condition does not necessarily hold even in the average case with respect to this stochastic behavior,
Alice and Bob might be able to efficiently generate secure keys.
In this case, Alice and Bob need to assign $a_E$ to the maximum value among possible values.
On the other hand, 
by random sampling, they can observe whether each coding block can generate secure keys.
Hence, there might be a possibility that a part of coding blocks can generate secure keys.
That is, they can apply the backward secure key distillation protocol
only to the coding blocks that can generate secure keys.
Such a selection of advantageous events to Alice and Bob
is called {\it post selection}.

\subsection{Choice of code and its calculation complexity}\Label{L89}
Now, we discuss the calculation complexity of our protocol.
In Protocol \ref{protocol1} except for Protocol \ref{protocol2},
we calculate only the averages of the obtained data and its square.
So, their calculation complexity is not so large.
Protocol \ref{protocol2} contains the calculation of syndrome,
the decoding of the given error correction code,
and universal2 hash function.
For Information reconciliation, 
we need to choose a suitable code, e.g. LDPC codes to satisfy the following condition.
\begin{description}
\item[(C1)]
The calculation of syndrome and
the decoding of the given error correction code
are have been already implemented with reasonable calculation complexity.
In fact, so many codes satisfy this condition \cite[p. 228]{RU}.
\item[(C2)]
The code can decode the message under the channel $W_{A|B}$, which is given in the Protocol 2.
\end{description}
If we use such a code, 
we can exploit existing algorithms for Step of Information reconciliation.

To achieve a larger key generation rate, 
we need to choose an error correction code $\CD \subset \bF_2^n$
whose coding rate is close to the capacity.
For this purpose, we employ an LPDC code with the brief propagation method, whose block length is around $2^{16} \cong 65,000$\cite[Chap. 4]{RU}.
However, we do not necessarily choose the block length of the error correcting code to be the block length $n$ of our protocol.
That is, we can consider the concatenation of our error correcting code.
When the block length of our error correcting code is $\frac{n}{k}$,
$k$ blocks of our error correcting code is treated as one block of our of our protocol, i.e.,
we apply one hash function to $k$ blocks of corrected keys of error correction.
That is, our LDPC code $\CD \subset \bF_2^{n/k}$ is chosen so that the dimension is less than 
$\frac{n}{k} I[\hat{P}_{A^c},\hat{c}_{AB}] $.
Since the agreement between Alice and Bob can be checked by error correction,
we do not need to evaluate the error of estimation $\hat{P}_{A^c}$.
That is, to decide the block length $n$ of our protocol,
we need to care about only the calculation complexity of hash function.

\subsection{Privacy amplification, verification, and their calculation complexity}
\subsubsection{Privacy amplification}
For Privacy amplification, we can use 
a modified Toeplitz matrix as a typical example of a universal2 hash function,
whose detail construction and evaluation of the complexity of its construction
are summarized in the recent paper \cite[Appendix]{H-T}.
Its calculation complexity is $O(m \log m)$ when $m$ is the input length.
Indeed, it was reported in paper \cite{H-T} that
the above type hash function practically implemented with $m=1 000 000$
by a conventional personal computer.
So, the part of privacy amplification
has only calculation complexity $O(n \log n)$.

Here, we need to calculate the size of the sacrifice bit length in privacy amplification.
This length should be chosen so that the security criterion \eqref{X1} or/and \eqref{X2}
is less than a given threshold, which shows the security level.
This calculation can be done by using the formulas \eqref{28-1} or/and \eqref{28-2} in Lemma \ref{LLO},
whose calculation complexity does not depend on the numbers of input and output lengths, as explained in Subsection \ref{s33}.
So, this process also can be done efficiently.

\subsubsection{Verification of correctness and public channel}\Label{NHS}
The error verification is also done by a universal2 hash function.
When we employ the above example,
it has only calculation complexity $O(m_2 \log m_2)$.
Due to this step, we can guarantee the correctness with probability $1-2^{-m_2}$, which is called the significance level\cite[Section VIII]{Fung}.
So, it is enough to choose $m_2$ depending on the required significance level.
Hence, 
we do not need to evaluate the decoding error probability for the step of information reconciliation.
That is, we do not need to care about the estimation error in the step of information reconciliation.
In contrast, we need to be careful for the estimation error in the step of privacy amplification because no method can evaluate the amount of information leaked to Eve
in the final keys without use of the estimation error.

Rigorously, in this protocol, 
Eve might override the signals to Bob or the public channel for spoofing \cite{Zeng,SC}.
To avoid Eve's spoofing, 
Alice and Bob needs verification of their public channel,
i.e., they need to authenticate each other \cite{Carter,WC81,KR94,KR95}.
Alice and Bob can authenticate each other by using universal2 hash function.
This authentication consumes a small number of secret keys between Alice and Bob.
Since the length of the keys for the authentication is smaller than the length of generated keys,
Alice and Bob can increase the length of the secret keys efficiently.
When we consume $k$ bits for the authentication for $n$-bit transmission, 
the authentication scheme is secure with a failure probability of
$ n 2^{-k+1}$ \cite[Theorem 9]{KR94}.
So, If Alice and/or Bob find disagreement, they consider that there exists spoofing and discard the obtained random variable.
Then, this protocol well works totally.

\section{Security analysis and sacrifice bit length}\Label{s33}

\subsection{Finite-length case with known parameters and distribution}\Label{SBF}
We analyze the security when the distribution of $P_Y$ and the parameters 
$a_B$, $a_E$, $v_Y$, $b_B$, $b_E$ and $e_B$ are known to Alice and Bob.
Since Eve knows the exchanged information via public communication,
she knows $\hat{e}_B $, i.e., she knows $e_B- \hat{e}_B $.
\subsubsection{Single-system description}\Label{SBG}
First, we discuss this problem with the single-system description,
in which $B'$ is defined as $B':=\sgn (B-\hat{e}_B)$.
So, due to a similar analysis to Theorem \ref{T2X},
Eve's knowledge for Bob's random variable
$B- \hat{e}_B 
=E' + \frac{a_B b_E}{a_E^2+b_E^2}U+ {b_B} X_1+e_B- \hat{e}_B $
can be reduced to $E'$
because $E'$ is independent of $\frac{a_B b_E}{a_E^2+b_E^2}U+ {b_B} X_1$.

In this section, we derive general security formulas
by using the true probability density function $P_{E'}$ of the random variable $E'$.
For this purpose, we introduce the functions 
$H[P_{E'},v]$ and $\phi[P_{E'},v](t)$ as
\begin{align}
& H[P_{E'},v]
\nonumber\\
:=& - \int_{-\infty}^\infty 
\Big[\Phi_1(\frac{x}{\sqrt{v}}) \log \Phi_1(\frac{x}{\sqrt{v}}) d x \nonumber \\
&\quad + (1-\Phi_1(\frac{x}{\sqrt{v}})) \log (1-\Phi_1(\frac{x}{\sqrt{v}}) )\Big]
P_{E'}(x)
d x
\Label{4-29-1},
\end{align}
and
\begin{align}
& \phi[P_{E'},v](t) \nonumber\\
:= & 
\log 
\int_{-\infty}^\infty 
(\Phi (\frac{x}{\sqrt{v}})^{\frac{1}{1-t}}
+
(1-\Phi (\frac{x}{\sqrt{v}}))^{\frac{1}{1-t}}
)^{1-t} 
P_{E'}(x)
dx, \Label{4-29-2B}
\end{align}
where the base of the logarithm is chosen to be $2$ in this paper.

Using the parameter $v$, we introduce the joint probability density function $P_{B',E'}$ 
of the random variables $B'$ and $E'$ as 
\begin{align}
P_{B',E'}(0,x)= \Phi_1(\frac{x}{\sqrt{v}}) P_{E'}(x), \quad
P_{B',E'}(1,x)= (1-\Phi_1(\frac{x}{\sqrt{v}})) P_{E'}(x).
\end{align}
So, the functions $H[P_{E'},v]$ and $\phi[P_{E'},v](t)$ are rewritten as
\begin{align}
 H[P_{E'},v] &=H(B'|E')[P_{B',E'}] \Label{4-29-3} \\
\phi[P_{E'},v] &=\phi(B'|E')[P_{B',E'}],\Label{4-29-2} 
\end{align}
where
\begin{align}
& H( B'|E')[P_{B',E'}]  
\nonumber\\
:=&
\int_{-\infty}^\infty P_{B',E'}(0,x) \log 
\frac{P_{B',E'}(0,x)+P_{B',E'}(1,x)}{P_{B',E'}(0,x)}
dx \nonumber\\
&+
\! \int_{-\infty}^\infty \! P_{B',E'}(1,x) \log 
\frac{P_{B',E'}(0,x)+P_{B',E'}(1,x)}{P_{B',E'}(1,x)}
dx ,\\
& 
\phi(B'|E')[P_{B',E'}](t) \nonumber\\
:= &
\log \int_{-\infty}^\infty (P_{B',E'}(0,x)^{\frac{1}{1-t}}
+P_{B',E'}(1,x)^{\frac{1}{1-t}})^{1-t} dx .
\end{align}
Here, this definition can be applied to a general pair of a binary variable $B'$ and a real variable $E'$.
Notice that we have $-\frac{d}{ds}\phi[P_{E'},v](s)|_{s=0}=H[P_{E'},v]$.
Using the property of this type of information quantity
given in \cite{Gal},
we have the following properties for the function $\phi[P_{E'},v](t)$.
\begin{lem}\Label{L30}
The function $\phi[P_{E'},v](t)$ is convex for $t \in (0,1)$.
\end{lem}
This lemma is shown in Subsection \ref{s53}.
Since the limit $\lim_{t\to 0}\frac{-\phi[P_{E'},v](t)}{t}$ equals the conditional entropy, 
$ H[P_{E'},v]$ is monotone decreasing for $\rho$.

\subsubsection{Multiple-system description}
To evaluate the security, we discuss the multiple-system description, in which
the variables $E_i,E_i', E_i'',A^c_i, U_i$  are defined in the same way as $A_i$ and $B_i$.
All of Eve's knowledge is written as ${\cal E}$.
As shown below, in the privacy amplification, we need to choose the sacrifice but rate $\frac{m_1}{n} $ is larger than
$ 1- H[P_{E'},v_{B|E'}]$, where $n$ is the block length before information reconciliation and $v_{B|E'}$ was defined to be the variance $\frac{a_B^2 b_E^2}{a_E^2+b_E^2}+b_B^2$.
To show this fact, we make more precise analysis on the leaked information as follows.
Using the relative entropy $D(P\|Q):=\sum_{x}P(x)(\log P(x)-\log Q(x))$
and the variational distance $d(P,Q):=\sum_x |P(x)-Q(x)|$,
we adopt the conditional modified mutual information $I'(K:{\cal E}|H)$ \cite{CN04,H-ep} between Bob and Eve
and the variational distance measure $d(K:{\cal E}|H)$ \cite{Renner} conditioned with $H$ as
\begin{align}
I'(K:{\cal E} |H) 
:=& \sum_{h} P_H(h) D(P_{{\cal E} K|H=h}\|P_{{\cal E}|H=h}\times P_{{\rm Uni},K}) \Label{X1} \\
d(K:{\cal E}|H) 
:=& \sum_{h} P_H(h) d(P_{{\cal E} K|H=h},P_{{\cal E}|H=h}\times P_{{\rm Uni},K}),\Label{X2}
\end{align}
where 
$P_{{\rm Uni},K} $ is the uniform distribution for the final key.
It is known that the latter satisfies the universal composable property \cite{R-K}.
Remember hat $H$ is the random variable to describe the choice of hash function.

Here, we give security formulas with true parameters.
So, the discussions in \cite{H-tight,H-leaked,H-ep,H-cq} yield the following lemma,
whose detail derivations are available in Subsection \ref{s53}.
\begin{lem}\Label{LLO}
We have
\begin{align}
I'(K: {\cal E}|H) \le & 
\inf_{s \in (0,1) }
\frac{1}{s}
2^{ s(n -m_1)+ n\phi[P_{E'},v_{B|E'}](s)} ,
\Label{28-1}\\
d(K:{\cal E}|H) \le &
3 \min_{t \in [0,\frac{1}{2}]} 2^{t(n-m_1)+ n\phi[P_{E'},v_{B|E'}](t)} .
\Label{28-2}
\end{align}
\end{lem}

Since the function $t\mapsto (n-m_1)+ n\phi[P_{E'},v_{B|E'}](t)$ is convex (Lemma \ref{L30}),
the minimum $\min_{t \in [0,\frac{1}{2}]} t(n-m_1)+ n\phi[P_{E'},v_{B|E'}](t) $
is computable by the bisection method \cite[Algorithm 4.1]{BV}, which gives the RHS of \eqref{28-2}.
Since $s \mapsto -\log  s$ is convex,
the function $s \mapsto   (s(n-m_1)+ n\phi[P_{E'},v_{B|E'}](s))-\log s$ is convex.
So, the infimum $\inf_{t \in (0,1)} t(n-m_1)+n \phi[P_{E'},v_{B|E'}](t) $
is computable in the same way. That is, we can calculate the RHS of \eqref{28-1} in Lemma \ref{LLO}.
When the sacrifice bit length $ \frac{m_1}{n}$
is greater than $ 1- H[P_{E'},v_{B|E'}] $,
there exists $s \in (0,\frac{1}{2}]$ such that
$(s(1-\frac{m_1}{n})+ \phi[P_{E'},v_{B|E'}](s))<0$.
So, both upper bounds go to zero exponentially for $n$.


Our condition for the random hash function $f_H$ can be relaxed to $\epsilon$-almost universal dual hash function \cite{Tsuru}.
(\cite{H-ep} contains its survey with non-quantum terminology.)
The latter class allows more efficient random hash functions with less random seeds \cite{H-T}.
Even when the random seeds $H$ is not uniform random number,
we have similar evaluations by attaching the discussion in \cite{H-T}.
While it is possible to apply left over hashing lemma \cite{BBCM,HILL} and smoothing to the min entropy \cite{Renner},
our evaluation is better than such a combination even in the asymptotic limit,
as is discussed in \cite{H-tight,H-ep}.

\subsection{Asymptotic case}
\subsubsection{Asymptotic case with known parameters and distribution}
Next, from the theoretical viewpoint,
we discuss the asymptotically achievable rate when the distribution of $P_Y$ and the parameters 
$a_B$, $a_E$, $v_Y$, $b_B$, $b_E$, and $e_B$ are known to Alice and Bob.
For simplicity, we consider the case when 
the information reconciliation 
asymptotically generates agreed keys between Alice and Bob
with the mutual information rate
$I(B';A)= 1-H( B'|A)[P_{B',A}]  $.
Since our focus in this section is limited to the asymptotic analysis
with independent and identical distributed setting,
it is sufficient to discuss the single-system description as Subsection \ref{SBG}.
The analysis in Section \ref{SBF} guarantees that 
the rate $ 1- H[P_{E'},v_{B|E'}]$ is asymptotically sufficient for the rate of of sacrificed keys.
Hence, the above method yields the asymptotic key generation rate 
$(1-H( B'|A)[P_{B',A}]  -(1- H[P_{E'},v_{B|E'}]))_+
=( H[P_{E'},v_{B|E'}]-H( B'|A)[P_{B',A}])_+$, where
$(x)_+:=\max(x,0)$.
Using the existing results of secure key generation \cite{AC,Ma,CN04}, we have the following theorem.

\begin{thm}\Label{TY}
(i) When there exists a cumulative distribution function $F$ such that
the distribution corresponding to $F$ is not the delta measure and 
\begin{align}
\Phi_{\frac{a_B^2 b_E^2}{a_E^2+b_E^2}}
=F_{Y}* F,
\Label{MML}
\end{align}
we have $H[P_{E'},v_{B|E'}]-H( B'|A)[P_{B',A}]> 0$, and
there is no protocol 
to generate secure keys between Alice and Bob from
the sequence of $B'$ and $A$
whose asymptotic key generation is greater than
$H[P_{E'},v_{B|E'}]-H( B'|A)[P_{B',A}]$.

(ii) Conversely, when there exists a cumulative distribution function $F$ such that
\begin{align}
F_{Y}= \Phi_{\frac{a_B^2 b_E^2}{a_E^2+b_E^2}}* F,\Label{MML2}
\end{align}
Alice and Bob cannot distill secure key from $A$ and $B'$.
\end{thm}

\begin{proof}
First, we show (i).
As mentioned in the previous section, 
Eve's information for $B- \hat{e}_B $
is summarized to $E'$.
Let $V_1$ be a random variable such that the cumulative distribution function is $F$
and it is independent of $E'$.
Condition \eqref{MML} guarantees 
\begin{align*}
F*F_{E'}
=F* (\Phi_{\frac{a_B^2 a_E^2}{a_E^2+b_E^2}} * F_{Y})
=\Phi_{\frac{a_B^2 a_E^2}{a_E^2+b_E^2}} * (F* F_{Y})
=\Phi_{\frac{a_B^2 a_E^2}{a_E^2+b_E^2}} * \Phi_{\frac{a_B^2 b_E^2}{a_E^2+b_E^2}}
= \Phi_{a_B^2 },
\end{align*}
which implies that 
the variable $E'+V_1$ is subject to the same distribution as $a_B A$.
Let $V_2$ be a random variable independent of $E'$ and $V_1$ that is subject to the same distribution as $Y+b_BX_1$.
Then, the joint distribution $P_{E'+V_1,E'+V_1+V_2}$
between $E'+V_1$ and $E'+V_1+V_2$ is the same as
the joint distribution $P_{a_B A, B}$ between $a_B A$ and $ B$.
Similarly, the joint distribution $P_{E',E'+V_1+V_2}$
between $E'$ and $E'+V_1+V_2$ is the same as
the joint distribution $P_{E', B}$  between $E'$ and $ B$.
Thus, we stochastically have the Markovian chain $E'\mc  A\mc B \mc B'$.
Hence, we have
$I(B';A|E')=
\sum_{e'}P_{E'}(e') 
(H(  \sgn(e'+V_1+V_2-\hat{e}_B) )
-\int_{-\infty}^{\infty} H(  \sgn(e'+v_1+V_2-\hat{e}_B) ) P_{V_1}(d v_1))$.
Since the variable $V_2$ takes values from $-\infty$ to $\infty$ with non-zero probability,
and $ P_{V_1}$ is not a delta measure,
we have
$H(  \sgn(e'+V_1+V_2-\hat{e}_B) )
-\int_{-\infty}^{\infty} H(  \sgn(e'+v_1+V_2-\hat{e}_B) ) P_{V_1}(d v_1)>0$
for any real number $e'$, which implies that $I(B';A|E')>0 $.
Therefore, 
the optimal key generation rate is calculated as \cite{AC,Ma,CN04}
\begin{align}
0< I(B';A|E')=
I(B';A)-I(B';E')=H[P_{E'},v_{B|E'}]-H( B'|A)[P_{B',A}].\Label{FOU}
\end{align}

Next, we show (ii).
When \eqref{MML2} holds, we have
\begin{align*}
F*\Phi_{a_B^2 }
=(F*\Phi_{\frac{a_B^2 b_E^2}{a_E^2+b_E^2}} )* \Phi_{\frac{a_B^2 a_E^2}{a_E^2+b_E^2}}
=F_{Y}*\Phi_{\frac{a_B^2 a_E^2}{a_E^2+b_E^2}} 
=F_{E'}.
\end{align*}
Let $V_3$ be a random variable such that the cumulative distribution function is $F$
and it is independent of $a_B A$.
Let $V_4$ be a random variable independent of $a_B A$ and $V_3$ that is subject to 
the Gaussian distribution with average 0 and variance 
$ \frac{a_B^2 b_E^2}{a_E^2+b_E^2}+ b_B^2$. 
Then, the joint distribution $P_{a_B A +V_3,a_B A +V_3+V_4}$
between $a_B A +V_3$ and $a_B A +V_3+V_4$ is the same as
the joint distribution $P_{E', B}$ between $E'$ and $ B$.
Similarly, the joint distribution $P_{a_B A ,a_B A +V_3+V_4}$
between $a_B A $ and $a_B A +V_3+V_4$ is the same as
the joint distribution $P_{a_B A, B}$  between $a_B A$ and $ B$.
Thus, we stochastically have the Markovian chain $ A\mc E'\mc B \mc B'$.
Hence, it is impossible to distill secure keys from $A$ and $B'$ \cite{AC,Ma,CN04}.
\end{proof}

To discuss more detail, we assume that $Y$ is subject to a Gaussian distribution with variance $v_Y$.
For the analysis under this assumption, we define
\begin{align}
& H(v)
\nonumber\\
:=& - \frac{1}{\sqrt{2\pi}}
\int_{-\infty}^\infty 
\Big[\Phi_1(\frac{x}{\sqrt{v}}) \log \Phi_1(\frac{x}{\sqrt{v}}) d x \nonumber \\
&\quad + (1-\Phi_1(\frac{x}{\sqrt{v}})) \log (1-\Phi_1(\frac{x}{\sqrt{v}}) )\Big]
e^{-x^2/2}
d x
\Label{4-29-T}.
\end{align}
$H(v)$ is strictly increasing for $v$.
Since
\begin{align}
& H(v_1/v_2)
\nonumber\\
=& - \frac{1}{\sqrt{2\pi v_2}}
\int_{-\infty}^\infty 
\Big[\Phi_1(\frac{x}{\sqrt{v_1}}) \log \Phi_1(\frac{x}{\sqrt{v_1}}) d x \nonumber \\
&\quad + (1-\Phi_1(\frac{x}{\sqrt{v_1}})) \log (1-\Phi_1(\frac{x}{\sqrt{v_1}}) )\Big]
e^{-x^2/2v_2}
d x
\Label{40-T},
\end{align}
we have $H[P_{E'},v_{B|E'}]
=H \Big(\frac{\frac{a_B^2 b_E^2}{a_E^2+b_E^2}+b_B^2}
{\frac{a_B^2 a_E^2}{a_E^2+b_E^2}+v_Y}\Big)
=H \Big(\frac{\frac{ \alpha_B }{ \alpha_E+1}+1}
{\frac{ \alpha_B \alpha_E}{\alpha_E+1}+v_Y'}\Big)
$ and 
$H( B'|A)[P_{B',A}]=H(\frac{v_Y+b_B^2}{a_B^2})
=H(\frac{ v_Y'+1}{\alpha_B})$,
where 
$\alpha_B:= \frac{a_B^2}{b_B^2}$,
$\alpha_E:= \frac{a_E^2}{b_E^2}$,
and $v_Y':= \frac{v_Y}{b_B^2}$.
Hence, we have the key generation rate
$(H[P_{E'},v_{B|E'}]-H( B'|A)[P_{B',A}])_+=
R_1:=\Big(
H \Big(\frac{\frac{ \alpha_B }{ \alpha_E+1}+1}
{\frac{ \alpha_B \alpha_E}{\alpha_E+1}+v_Y'}\Big)
-
H(\frac{ v_Y'+1}{\alpha_B}) \Big)_+$.
This value is strictly positive if and only if 
$ v_Y'< \frac{\alpha_B }{\alpha_E+1}$, i.e., 
$v_Y< \frac{a_B^2 b_E^2}{a_E^2+b_E^2}$
because $H(v)$ is strictly increasing for $v$.
Notice that the condition 
$ v_Y'< (\ge) \frac{\alpha_B }{\alpha_E+1}$, i.e., 
$v_Y <  (\ge) \frac{a_B^2 b_E^2}{a_E^2+b_E^2}$ 
equivalent to the condition of Part (i) ((ii)) 
of Theorem \ref{TY} in this case.

\subsubsection{Asymptotic case with estimation}
We consider the case when 
the distribution of $P_Y$ and the parameters 
$a_B$, $a_E$, $v_Y$, $b_B$, $b_E$, and $e_B$ are known to Alice and Bob.
We assume the asymptotic case, in which $l$ goes to infinity but $\frac{l}{n}$ goes to zero.
In this case,
they estimate these parameters by Steps 1 and 2 pf Protocol 1.
The estimation errors for these parameters go to zero.
Also, as discussed in Section \ref{ESDIS},
the estimation error for the distribution of $P_Y$ goes to zero.
Therefore, they can achieve the asymptotic key generation rate 
$H[P_{E'},v_{B|E'}]-H( B'|A)[P_{B',A}]$.

\subsubsection{Comparison with one-way case}\Label{MDO}
To compare our protocol with the one-way wire-tap channel,
we assume that 
the distribution of $P_Y$ and the parameters 
$a_B$, $a_E$, $v_Y$, $b_B$, $b_E$, and $e_B$ are known to Alice, Bob, and Eve.
In fact, 
the following modified protocol can be reduced to the one-way case.
In Step 1 of Protocol 2, Alice makes the random variable $A_i':=\sgn A_i$.
Then, we make information reconciliation (Step 2 of Protocol 2) with the opposite direction
(Alice sends the syndrome to Bob via public channel).
This modification is called forward reconciliation.
Sending Alice's syndrome is equivalent to restricting Alice's variables to a special coset with respect to $C$.
Hence, the analysis with forward reconciliation can be reduced to the one-way wire-tap channel.
Then, using existing results of wire-tap channel \cite{CK79}, we have the following lemma.

\begin{lem}\Label{TY2}
(i) When $\frac{a_B}{b_B} >\frac{a_E}{b_E}$
and there exists a cumulative distribution function $F$ such that
the distribution corresponding to $F$ is not the delta measure and 
\begin{align}
\Phi_{\frac{b_E^2 a_B^2}{a_E^2}-b_B^2}
=F_{Y}* F,\Label{MML3}
\end{align}
there exists  a secure wire-tap code for the wire-tap channel \eqref{4-24-1G} and \eqref{4-24-2G}. 

(ii) Conversely, assume that 
$\frac{a_B}{b_B} <\frac{a_E}{b_E}$
or there exists a cumulative distribution function $F$ such that
\begin{align}
F_{Y}= \Phi_{
\frac{b_E^2 a_B^2}{a_E^2}-b_B^2}* F,\Label{MML4}
\end{align}
the wire-tap channel \eqref{4-24-1G} and \eqref{4-24-2G} cannot transmit secure information from Alice to Bob.
\end{lem}

\begin{proof}
First, we show (i).
Let $V_1$ be a random variable such that the cumulative distribution function is $F$
and it is independent of $E'$.
Condition \eqref{MML} guarantees 
\begin{align*}
F*F_{Y+b_B X_1}
=F* (\Phi_{b_B^2} * F_{Y})
=\Phi_{b_B^2} * (F* F_{Y})
=\Phi_{b_B^2} * \Phi_{\frac{b_E^2 a_B^2}{a_E^2}-b_B^2}
= \Phi_{\frac{b_E^2 a_B^2}{a_E^2}},
\end{align*}
which implies that 
the variable $Y+b_B X_1+V_1$ is subject to the same distribution as $\frac{b_E a_B}{a_E} X_2$.
Therefore, 
the channel from $A$ to $\frac{a_E}{a_B}(B+V_1)$ has the same conditional distribution as that of 
the channel from $A$ to $E$.
Hence, the channel from Alice to Eve can be regarded as a degraded channel of 
the channel from Alice to Bob.
Hence, 
the capacity is given as 
the maximum of 
\begin{align}
&I(A;B)-I(A;E)= H(A|E)-H(A|B)= H(A|E)-H(A|BE) \nonumber \\
=&
\int_{-\infty}^{\infty} 
\Big(
H(A|E=e)- \int_{-\infty}^{\infty}H(A|B=b) 
P_{B|E=e}(db) \Big)P_E(de)\Label{BTR}
\end{align}
with respect to the choice of the distribution of $A$ \cite{CK79}, 
In this case, 
the conditional distribution 
$P_{A|B=b}$ is different from 
$P_{A|B=b'}$ when $b\neq b'$.
The assumption for $F$ guarantees that $V_1$ is not a deterministic value.
Hence, 
the conditional distribution 
$P_{B|E=e}$ has probability at least two points.
Hence, we find that 
$H(A|E=e)- \int_{-\infty}^{\infty}H(A|B=b) 
P_{B|E=e}(db) >0$.
Therefore, the value \eqref{BTR} is strictly positive.
Hence, we obtain the statement (i).

Next, we show (ii).
Let $V_2$ be a random variable such that the cumulative distribution function is $F$
and it is independent of $X_2$.
Then, we have
\begin{align*}
F_{V_2}*F_{\frac{b_E a_B}{a_E} X_2}
=F*\Phi_{\frac{b_E^2 a_B^2}{a_E^2}}
=( F* \Phi_{\frac{b_E^2 a_B^2}{a_E^2}-b_B^2})* \Phi_{b_B^2}
=F_{Y}* \Phi_{b_B^2}
\end{align*}
which implies that 
the variable $\frac{b_E a_B}{a_E} X_2+V_1$ is subject to the same distribution as $Y+b_B X_1 $.
Therefore, 
the channel from $A$ to $\frac{b_E a_B}{a_E} E +V_2$ has the same conditional distribution as that of 
the channel from $A$ to $B$.
Hence, the channel from Alice to Bob can be regarded as a degraded channel of 
the channel from Alice to Eve.
Thus, the wire-tap channel \eqref{4-24-1G} and \eqref{4-24-2G} cannot transmit secure information from Alice to Bob \cite{CK79},
which is the statement (ii).
\end{proof}

Now, we compare the conditions of Part (i) of Theorem \ref{TY} and Lemma \ref{TY2}.
When $\frac{a_B^2 b_E^2}{a_E^2+b_E^2}>\frac{b_E^2 a_B^2}{a_E^2}-b_B^2 $,
the condition of Part (i) of Theorem \ref{TY} 
is weaker than the condition of Part (i) of Theorem \ref{TY2}. 
Since 
\begin{align}
\frac{a_B^2 b_E^2}{a_E^2+b_E^2}-\Big(\frac{b_E^2 a_B^2}{a_E^2}-b_B^2 \Big)
=
\frac{b_E^4 b_B^2}{a_E^2(a_E^2+b_E^2)} 
\Big( (\frac{a_E}{b_E})^4 +(\frac{a_E}{b_E})^2-(\frac{a_B}{b_B})^2\Big)
=
\frac{b_E^4 b_B^2}{a_E^2(a_E^2+b_E^2)} 
\Big( \alpha_E^2 +\alpha_E-\alpha_B \Big),
\end{align}
this condition is equivalent to 
\begin{align}
\alpha_E^2 +\alpha_E>\alpha_B.
\Label{MMD}
\end{align}
Since $\alpha_E$ and $\alpha_B$
are the ratios between the signal power and the power of detector noise 
of Eve and Bob, respectively, it is natural to assume that 
this ratio of Eve is equal to or larger that that of Bob, which implies
the inequality \eqref{MMD}.
Hence, when this inequality holds,
the condition of Part (i) of Theorem \ref{TY} 
is weaker than the condition of Part (i) of Lemma \ref{TY2}. 
That is, our method has a higher possibility to generate secure keys.

To discuss more details, we assume that $Y$ is subject to a Gaussian distribution with variance $v_Y$.
Then, Theorem \ref{TY} shows that 
Alice and Bob can distill secure keys from $A$ and $B'$ by using our method
 if and only if 
\begin{align}
\frac{a_B^2 b_E^2}{a_E^2+b_E^2} > v_Y, \hbox{ i.e., }
\frac{\alpha_B }{\alpha_E+1}> v_Y'.\Label{LOG2}
\end{align}
Lemma \ref{TY2} shows that Alice can send secure information to Bob via the one-way protocol based on the wire-tap channel \eqref{4-24-1G} and \eqref{4-24-2G}
if and only if 
\begin{align}
\frac{b_E^2 a_B^2}{a_E^2}-b_B^2  > v_Y, \hbox{ i.e., }
\frac{\alpha_B}{\alpha_E}-1  > v_Y'. \Label{LOG}
\end{align}
Therefore, under the natural condition \eqref{MMD},
our method has weaker condition to distill secure keys than the condition for secure communication in the one-way protocol.
In fact, in the natural setting, 
Eve's ratio $\alpha_E$ is equal to or larger that Bob's ratio $\alpha_B$, which implies
that the condition \eqref{LOG} does not holds for any $v_Y'$.
However, even under this case, 
we have a possibility to satisfy the condition \eqref{LOG2} for our method.

To discuss the detail, we consider the secure capacity for the wire-tap channel \eqref{4-24-1G} and \eqref{4-24-2G}
under the energy constraint .
When Alice the wire-tap channel $n$ times, we denote the set of codewords by ${\cal M}_n \subset \mathbb{R}^n$.
For ${\cal M}_n$, 
we impose the condition $ \sum_{i=1}^n x_i^2\le n$ with any sequence $(x_1,\ldots, x_n)\in {\cal M}_n $.
Under this condition, the secure capacity is calculated as \cite{LH}
\begin{align}
\Big(\frac{1}{2}\log (1+\frac{ a_B^2}{ v_Y+b_B^2})
-\frac{1}{2}\log (1+\frac{a_E^2}{b_E^2}) \Big)_+\nonumber  \\
=\Big(\frac{1}{2}\log (1+\frac{ \alpha_B }{ v_Y'+1})
-\frac{1}{2}\log (1+\alpha_E) \Big)_+.
\end{align}
With this rate, the strong security also holds \cite[Appendix D-C]{HM}.
This quantity is strictly positive if and only if the condition \eqref{LOG} holds.

\subsection{Estimation with confidence level in finite-length case}
\subsubsection{Estimation of parameters}
To estimate the average $e_B$, the variance $v_B$,
the covariance $c_{AB}$ between $A$ and $B$, and the distribution $P_{E'}$, 
we set the confidence level $1-\epsilon$.
Due to the quasi static assumption,
the random variables 
$A$, $B$, and $E'$ are is subject to an identical and independent distribution.
Since the distribution of $B$ is unknown,
if the number $l$ of samples is not so large,
it is not easy to give the confidence interval for the estimation of $e_B$.
However, when the number $l$ is sufficiently large (e.g., more than $10^4$),
it is allowed to apply Gaussian approximation for a given confidence level $1-\epsilon$.
When the variance is unknown, we need to employ the $t$-distribution of degree $l-1$.
However, since the number $l$ is sufficiently large,
it can be well approximated by the Gaussian distribution.
Now, we use the $\epsilon$ percent point $Z_\epsilon$ (the quantile) of 
the standard Gaussian distribution.
The confidence interval of the average $e_B$ is 
$[\hat{e}_B- \sqrt{\bar{v}_B} Z_\epsilon l^{-\frac{1}{2}},
\hat{e}_B+ \sqrt{\bar{v}_B} Z_\epsilon l^{-\frac{1}{2}}]$
by using the sample mean $\hat{e}_B$ and the unbiased variance $\bar{v}_B$.
However, due to the largeness of $l$, 
the unbiased variance $\bar{v}_B$ can be replaced by the sample variance
because the difference is almost negligible.

Next, we estimate the variance $v_B$.
When $B$ is subject to the Gaussian distribution,
we need to employ the $\chi^2$ distribution of degree $l-1$
unless the number $l$ is sufficiently large.
Now, we can apply Gaussian approximation because the number $l$ is sufficiently large.
To estimate the variance of $(B-e_B)^2$, 
we define the estimate 
$\bar{w}_B
:=\frac{1}{l-1}\sum_{i=1}^l 
((\bar{B}_i-\hat{e}_B)^2 
- \frac{1}{l-1}\sum_{i=1}^l (\bar{B}_j-\hat{e}_B)^2)^2$,
which approximates the variance of $(B-e_B)^2$.
The confidence interval of the variance $v_B$ is 
$[\bar{v}_B- \sqrt{\bar{w}_B} Z_\epsilon l^{-\frac{1}{2}},
\bar{v}_B+ \sqrt{\bar{w}_B} Z_\epsilon l^{-\frac{1}{2}}]$.

Now, we estimate the covariance $c_{AB}$ between $A$ and $B$ by using 
the sample mean $\hat{c}_{AB}$ of $A(B-\hat{e}_B)$.
To get the confidence interval of the covariance $c_{AB}$,
we employ the unbiased variance $\hat{v}_{AB}$ of 
$(A-\hat{e}_A)(B-\hat{e}_B)$, which approximates the variance of 
the sample mean $\hat{c}_{AB}$.
So, the confidence interval of the covariance $c_{AB}$ is 
$[\bar{c}_{AB}- \sqrt{\bar{v}_{AB}} Z_\epsilon l^{-\frac{1}{2}},
\bar{c}_{AB} + \sqrt{\bar{v}_{AB}} Z_\epsilon l^{-\frac{1}{2}}]$,
where
$\bar{c}_{AB}:=\frac{1}{l-1}\sum_{i=1}^l (\bar{A}_i-\hat{e}_A)(\bar{B}_i-\hat{e}_B)$.

\subsubsection{Estimation of distribution}\Label{ESDIS}
To get the estimate $\hat{P}_{E'}$, 
we define the random variable $\bar{A}^c:=B- \hat{c}_{AB} A- \hat{e}_B$.
Using the second $l$ data $(\tilde{A}_1,\tilde{B}_1), \ldots, 
(\tilde{A}_l,\tilde{B}_l)$, 
we define our estimate
\begin{align}
\hat{F}_{\bar{A}^c}(x)&:={1 \over l}\sum_{i=1}^l I_{[-\infty,x]}(\bar{A}^b_i)
\end{align}
based on Kolmogorov-Smirnov test \cite{Kolmogorov,Smirnov},
where
$\bar{A}^b_i:=\tilde{B}_i- \hat{c}_{AB} \tilde{A}_i- \hat{e}_B$

When $ \frac{\hat{c}_{AB}^2 a_E^2}{a_E^2+b_E^2}\ge b_B^2$,
we define our estimate as
\begin{align}
\hat{F}_{E'}:=
{\cal G}_{\frac{\hat{c}_{AB}^2 a_E^2}{a_E^2+b_E^2}- b_B^2}[\hat{F}_{\bar{A}^c}].
\end{align}
Otherwise,
we define our estimate as $\hat{F}_{E''}:=\hat{F}_{\bar{A}^c}$.
The estimate $\hat{P}_{E'}$ ($\hat{P}_{E''}$) of the distribution $P_{E'}$ ($P_{E''}$)
is given as the derivative of $\hat{F}_{E'}$ ($\hat{F}_{E''}$).


Now, we evaluate the error of these estimates $\hat{F}_{E'}$ and $\hat{F}_{E''}$.
To estimate the error of this estimator,
we define the Kolmogorov distribution function $L(x)$\cite{MTW};
\begin{align}
L(x):=1-2\sum_{k=1}^\infty (-1)^{k-1} e^{-2k^2 x^2}=\frac{\sqrt{2\pi}}{x}\sum_{k=1}^\infty e^{-(2k-1)^2\pi^2/(8x^2)},
\end{align}
which can also be expressed by the Jacobi theta function 
$\displaystyle \vartheta _{01}(z=0;\tau =2ix^{2}/\pi )$.
Then, we have the following lemma with an integer $l_{\epsilon,\delta}$,
which is defined later for two real numbers $\epsilon,\delta>0$.
This lemma will be shown in Subsection \ref{s11-17-1}.
\begin{lem}\Label{L11-10}
Under the condition $\frac{\hat{c}_{AB}^2 a_E^2}{a_E^2+b_E^2}\ge b_B^2$,
the estimate $\hat{F}_{E'}$ for ${F}_{E'}$ satisfies 
the inequality
\begin{align}
& \sup_x
|{F}_{E'}(x)
-\hat{F}_{E'}(x)| \nonumber \\
\le &
\frac{\sqrt{\hat{v}_{AB}}}{\sqrt{2\pi e}\hat{c}_{AB} \sqrt{l}}Z_{\epsilon}
+\frac{1}{\sqrt{l}} L^{-1}( 1-\epsilon)\Label{e11-10}
\end{align}
with confidence level $1-2(\epsilon+\delta)$ when $l\ge l_{\epsilon,\delta}$.
Under the condition $\frac{\hat{c}_{AB}^2 a_E^2}{a_E^2+b_E^2}< b_B^2$,
the estimate $\hat{F}_{E''}$ for ${F}_{E''}$ satisfies 
the inequality
\begin{align}
& \sup_x
|{F}_{E''}(x)
-\hat{F}_{E''}(x)|  \nonumber \\
\le &
\frac{\sqrt{\hat{v}_{AB}}}{\sqrt{2\pi e}\hat{c}_{AB} \sqrt{l}}Z_{\epsilon}
+\frac{1}{\sqrt{l}} L^{-1}( 1-\epsilon)\Label{e11-10-2}
\end{align}
with confidence level $1-2(\epsilon+\delta)$ when $l\ge l_{\epsilon,\delta}$.
\end{lem}

To show define $l_{\epsilon,\delta}$,
we employ Kolmogorov-Smirnov test \cite{Kolmogorov,Smirnov}, whose detail is the following.
We consider the independent random variables $\bar{X}_1,\ldots, \bar{X}_l$
subject to the distribution $P_X$, whose cumulative distribution function is $F_X(x)$.
Then, we define the empirical distribution function 
\begin{align}
F_{X,l}(x):={1 \over l}\sum_{i=1}^l I_{[-\infty,x]}(\bar{X}_i),
\end{align}
where $I_{[-\infty,x]}(X)$ is the indicator function, equal to 1 if $X \le x$ and equal to $0$ otherwise.
We define the random variable 
\begin{align}
D_{X,l}:= \sup_{x} |F_{X,l}(x) -F_X(x)|.
\end{align}
Then, we have the following lemma.

\begin{proposition}[Kolmogorov-Smirnov test\cite{MTW}]\Label{MLF}
The equation
\begin{align}
\lim_{l \to \infty} 
{\rm Pr} (D_{X,l} \le \frac{1}{\sqrt{l}} L^{-1}( 1-\epsilon))
= 1-\epsilon
\end{align}
holds.
\end{proposition}

For two real numbers $\epsilon,\delta>0$, we define the integer $l_{\epsilon,\delta}$ as
\begin{align}
l_{\epsilon,\delta}:=\min
\Big\{l' \Big|
{\rm Pr} (D_{X,l} \le \frac{1}{\sqrt{l}} L^{-1}( 1-\epsilon))
\le 1-(\epsilon+\delta) \hbox{ for } l\ge l' \Big\}.
\end{align}
The above Proposition guarantees the finiteness of $l_{\epsilon,\delta}$.
Thus, the estimate $\hat{F}_{\bar{A}^c}$ for $F_{\bar{A}^c}$ satisfies
the relation
\begin{align}
\sup_x
|F_{\bar{A}^c}(x)-\hat{F}_{\bar{A}^c}(x)| 
\le\frac{1}{\sqrt{l}}L^{-1}( 1-\epsilon)
\Label{eq11-19-1}
\end{align}
with confidence level $1-(\epsilon+\delta)$ when $l \ge l_{\epsilon,\delta}$.

\subsection{Security analysis with estimation in finite-length case}
In our protocol given in Section \ref{s31},
in addition to the choice of hash functions,
there are other random variables that are publicly transmitted.
For example, the information for error estimation is publicly transmitted between Alice and Bob.
So, the collection of them are denoted by $C$, and its distribution is denoted by $P_C$, which depends on the distribution $P_Y$ and the parameters $a_B$ and $e_B$.
In this case, the length of sacrifice bit length $m_1$ depends on $C$. So, it is denoted by $m_1(C)$.
The distribution $P_H$ of the choice of the hash function also depends $m_1(C)$. 
So, it is given by the conditional distribution $P_{H|m_1(C)} $.
Since the length of final keys also depends on $C$,
the uniform distribution is given by the conditional distribution $P_{{\rm Uni},K|C}$.
Thus, the security criteria \eqref{X1} and \eqref{X2} are modified to 
\begin{align}
I'(K: {\cal E}|H C) 
:=& \sum_{c,h} P_C(c) P_{H|m_1(C)=m_1(c)}(h) D(P_{{\cal E}K|H=h,C=c}\|P_{{\cal E}|H=h,C=c}\times P_{{\rm Uni},K|C=c}) \Label{X1b} \\
d(K:{\cal E}|H C) 
:=& \sum_{c,h} P_C(c) P_{H|m_1(C)=m_1(c)}(h) d(P_{{\cal E}K|H=h,C=c},P_{{\cal E}|H=h,C=c}\times P_{{\rm Uni},K|C=c}),\Label{X2b}
\end{align}
where $P_{{\rm Uni},K|C=c}$ is the uniform distribution of $K$ with the length determined by $C=c$.

For given public information $C$, we define 
\begin{align}
2^{\hat{\phi}(C,\epsilon)(t)}:=
\left\{
\begin{array}{ll}
2^{\phi[\hat{P}_{E'},\underline{v}_{B|E'}](t)}
+2(1- 2^{-t})
(\frac{\sqrt{\hat{v}_{AB}}}{\sqrt{2\pi e l}\hat{c}_{AB} }Z_{\epsilon}
+\frac{1}{\sqrt{l}} L^{-1}( 1-\epsilon) )
&
\hbox{when }
\frac{\hat{c}_{AB}^2 a_E^2}{a_E^2+b_E^2}\ge b_B^2
\\
2^{\phi[\hat{P}_{E''},\underline{c}_{AB}^2 ](t)}
+2(1- 2^{-t})
(\frac{\sqrt{\hat{v}_{AB}}}{\sqrt{2\pi e l}\hat{c}_{AB} }Z_{\epsilon}
+\frac{1}{\sqrt{l}} L^{-1}( 1-\epsilon) )
&
\hbox{otherwise,}
\end{array}
\right.
\end{align}
where
$\underline{v}_{B|E'}:=\frac{(\hat{c}_{AB}-\frac{\sqrt{\hat{v}_{AB}}}{ \sqrt{l}}Z_{\epsilon}
)^2 b_E^2}{a_E^2+b_E^2} +{b}_{B}^2$
and
$\underline{c}_{AB}:=\hat{c}_{AB}-\frac{\sqrt{\hat{v}_{AB}}}{ \sqrt{l}}Z_{\epsilon}$.

Hence, 
using Proposition \ref{MLF} (Kolmogorov-Smirnov test\cite{MTW}),
we obtain the following lemma and theorem, which will be shown in Subsection \ref{s11-17-2}. 
\begin{lem}\Label{L11-11b}
The function $\hat{\phi}(C,\epsilon)$ satisfies the inequality
\begin{align}
2^{\phi[P_{E'},v_{B|E'}](t)}
\le
2^{\hat{\phi}(C,\epsilon)(t)}\Label{11-20-3}
\end{align}
with confidence level 
$1-2(\epsilon+\delta)$ when $l\ge l_{\epsilon,\delta}$.
\end{lem}

Combining Lemmas \ref{LLO} and \ref{L11-11b} with 
Theorem \ref{T2X} and \eqref{E16D},
we obtain the following theorem.
\begin{thm}\Label{11-17T}
The function $\hat{\phi}(C,\epsilon)$ satisfies the inequality
\begin{align}
I'(K:{\cal E}|HC) \le & 
\inf_{s \in (0,1) }
\frac{1}{s}
2^{ s(n-m_1)+ n\hat{\phi}(C,\epsilon) (s)} ,
\Label{28-1X}\\
d(K:{\cal E}|HC) \le &
3 \min_{s \in [0,\frac{1}{2}]} 2^{s(n-m_1)
+ n\hat{\phi}(C,\epsilon)(s)} ,
\Label{28-2X}
\end{align}
with confidence level $1-2(\epsilon+\delta)$ when $l\ge l_{\epsilon,\delta}$.
\end{thm}
When we focus on the security criterion $I'(K:{\cal E}|HC)$, 
by using Theorem \ref{11-17T},
given a security level $\kappa$ and the observed values, 
the sacrifice bit length $m_1$ is chosen as
\begin{align}
\argmin \bigg\{m_1 \bigg| \inf_{s \in (0,1) }
\frac{1}{s} 2^{ s(n-m_1)+ n\hat{\phi}(C,\epsilon) (s)} \le \kappa\bigg\}.
\Label{28-1Y}
\end{align}
When we focus on the other security criterion $d(K:E|HC) $, 
the sacrifice bit length $m_1$ is chosen as
\begin{align}
\argmin \bigg\{m_1 \bigg| 
3 \min_{s \in [0,\frac{1}{2}]} 2^{s(n-m_1)
+ n\hat{\phi}(C,\epsilon)(s)} \le \kappa\bigg\}.
\Label{28-2Y}
\end{align}
That is, when we choose the sacrifice bit length $m_1$ based on \eqref{28-1Y} or \eqref{28-2Y},
the leaked information is less than $\kappa$ with confidence level $1-2(\epsilon+\delta)$ when $l\ge l_{\epsilon,\delta}$.

\subsection{Typical case}\Label{S5E}
To treat our model more concretely, in the following typical case,
we consider the key generation rate when the sacrifice bit length is decided by the formulas \eqref{28-1Y} and \eqref{28-2Y}. 
In this subsection, for simplicity, we discuss only the case when $Y$ is a Gaussian random variable and $e_B=\hat{e}_B=0$
while the employed formulas \eqref{28-1Y} and \eqref{28-2Y} do not assume that $Y$ is a Gaussian random variable.
First, we assume that there is no error between the true values and our estimations.
That is,
$\hat{e}_B =e_B$, $\hat{c}_{AB}=a_B$,
$\hat{v}_B =a_B^2+v_B+b_B^2$,
and $\hat{P}_{A^c}=P_{A^c}$.
By using the probability density function $\varphi$ of the standard Gaussian variable and the correlation coefficient,
the quantity $H [P_{\bar{E}},v_{B|E'}]$ can be written 
to be $H [\varphi,\frac{1-\rho_E^2}{\rho_E^2}]$
because $\frac{1-\rho_E^2}{\rho_E^2}= \frac{v_{B|E'}}{v_{E'}}$.
So, the required sacrifice bit rate is $1-H [\varphi,\frac{1-\rho_E^2}{\rho_E^2}]$,
which is mutual information between $E'$ and $B'$.
On the other hand, 
the mutual information $ I(\hat{P}_{A^c}, \hat{c}_{AB})$
between $A$ and $B'$ is calculated to be $1-H [\frac{1-\rho_A^2}{\rho_A^2}]$.
That is, the secure key generation rate is $H [\varphi,\frac{1-\rho_E^2}{\rho_E^2}]-H [\varphi,\frac{1-\rho_A^2}{\rho_A^2}]$ under the reverse information reconciliation.

Now, we consider the following special case.
Eve's detector has the same performance as Bob' detector, i.e.,
$b_E=b_B$, which will be denoted by $b$.
The coefficients $a_B$ and $a_E$ for attenuations equals the same value $\sqrt{2} b$.
By using the variance of the noise $Y$, 
the correlation coefficients $\rho_A$ and $\rho_E$ are calculated as 
$\rho_A^2= \frac{2}{3+\frac{v_Y}{b^2}}$
and $\rho_E^2= \frac{4+\frac{3v_Y}{b^2}}{9+\frac{3v_Y}{b^2}}$, i.e.,
$\frac{1-\rho_A^2}{\rho_A^2}=\frac{1+\frac{v_Y}{b^2}}{2}$
and 
$\frac{1-\rho_E^2}{\rho_E^2}=\frac{5}{4+\frac{3v_Y}{b^2}}$.
If the noise $Y$ generated in the transmission is not zero,
the mutual information between $A$ and $E$ is larger than that between $A$ and $B$.
So, the forward information reconciliation cannot generate any keys.
However, when we employ the reverse information reconciliation,
there is a possibility to generate secure keys.
When $v_Y < \frac{2b^2}{3}$, we have
$\rho_A^2 > \rho_E^2$, i.e., 
the secure key generation rate is the positive value
$H [\varphi,\frac{5}{4+\frac{3v_Y}{b^2}}]
-H [\varphi,\frac{1+\frac{v_Y}{b^2}}{2}]$
under the reverse information reconciliation,
which is numerically calculated as Fig. \ref{rate1}.
In particular, when $v_Y = \frac{b^2}{5}$,
the secure key generation rate 
$H [\varphi,\frac{5}{4+\frac{3v_Y}{b^2}}]
-H [\varphi,\frac{1+\frac{v_Y}{b^2}}{2}]$ is 0.108,
and the mutual informations $1-H [\varphi,\frac{1+\frac{v_Y}{b^2}}{2}]$
and 
$1-H [\varphi,\frac{5}{4+\frac{3v_Y}{b^2}}]$
are 0.372 and 0.264.
In this special case, 
the coding rate of error correcting code needs to be less than 0.372, 
and the sacrifice bit rate needs to be greater than and 0.264.

\begin{figure}[htbp]
\begin{center}
\scalebox{1}{\includegraphics[scale=1]{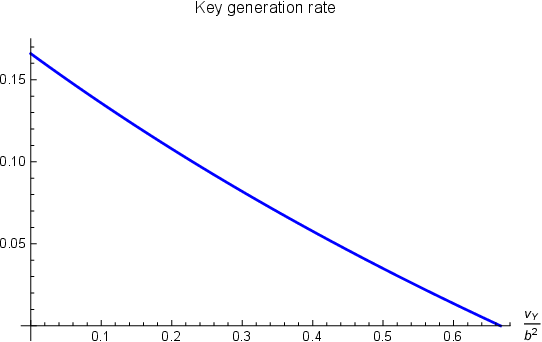}}
\end{center}
\caption{Key generation rate}
\Label{rate1}
\end{figure}%

However, when we care about the finiteness of the block length of our code,
the amount of leaked information of our final keys is not zero even though
the key generation rate is less than 
$H [\varphi,\frac{5}{4+\frac{3v_Y}{b^2}}]
-H [\varphi,\frac{1+\frac{v_Y}{b^2}}{2}]$.
To discuss this issue, we need to care about estimation error.
In the following, we assume the same assumption as the above discussion except for the relation between the true values and our estimations.
For this purpose, we briefly discuss the upper bounds \eqref{28-1X} and \eqref{28-2X} by taking account into estimation error.
To keep a high precision, we set the confidence level to be $1-10^{-4}$, i.e., $\epsilon=5 \times 10^{-5}$.
So, $Z_\epsilon=4.06$ and $L^{-1}(1-\epsilon)=2.30$ \cite{Ferguson}.
For example, we employ block length $n=1,000,000$. 
So, it is natural to choose the number of sampling to be the same value$1,000,000$, i.e., $l=500,000$.
In graph \ref{leaked1},
we numerically calculate the logarithm
$s(n-m_1) + n\hat{\phi}(C,\epsilon)(s)+\log 3$
of the upper bounds appeared in and \eqref{28-2X} as a function of $s$
when $v_Y = \frac{b^2}{5}$ and the sacrifice bit length $m_1$ is 
$0.30\times 1,000,000=300,000$.
The minimum value is $-867$ and is realized when $s=0.07$.
That is, when the required security level $\kappa$ is chosen to be $2^{-867}$,
the sacrifice bit length given in \eqref{28-2Y} is $300,000$.
Here, the logarithm $ s(n-m_1)+ n\hat{\phi}(C,\epsilon) (s) - \log s$
of the upper bounds appeared in \eqref{28-1X}
has almost the same behavior as that in \eqref{28-2X}.

In this numerical calculation, we need the value $\hat{v}_{AB}$.
When $Y$ is a Gaussian random variable,
the expectation of $\hat{v}_{AB}$ is 
$ 2c_{AB}^2+ v_B $, which equals 
$ 2a_B^2+ (a_B^2+b_B^2+v_Y)=7 b^2+v_Y$.
Notice that the variance of $A$ is $1$.
Also, $\frac{\underline{v}_{B|E'}}{v_{E'}}$
is 
$(\frac{(\sqrt{2}-\frac{\sqrt{\hat{v}_{AB}}}{ b\sqrt{l}}Z_{\epsilon}
)^2 }{3} +1)/(\frac{4}{3}+\frac{v_Y}{b^2})
=(\frac{(\sqrt{2}-\frac{\sqrt{7+\frac{v_Y}{b^2}}}{ \sqrt{l}}Z_{\epsilon}
)^2 }{3} +1)/(\frac{4}{3}+\frac{v_Y}{b^2})$.


\begin{figure}[htbp]
\begin{center}
\scalebox{1}{\includegraphics[scale=1]{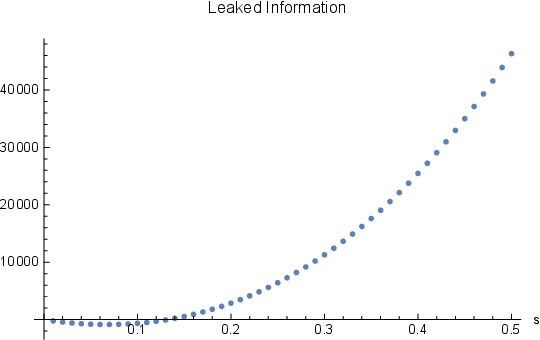}}
\end{center}
\caption{Logarithm of upper bound of leaked information}
\Label{leaked1}
\end{figure}%

\section{Extensions}\Label{s7}
\subsection{Multi-antenna attack}\Label{s23}
As a more powerful Eve, 
we assume that, 
instead of \eqref{4-24-2G},
Eve can prepare $k$ antennas that receiving
$E_j$ ($j=1, \ldots, k$) under Assumptions (A1)-(A6) as
\begin{align}
E_j=a_{E,j}A+b_{E,j}X_{2,j}.
\end{align}
Bob receives $B$ given in \eqref{4-24-1G}, and 
$X_{2,i}$ are subject to the standard Gaussian distribution independently of other random variables $X_i$ ($i=1,2$).

That is, Eve knows $E_1, \ldots, E_k$ as well as $Y$.
Now, we convert 
the random variables $E_1, \ldots, E_k$ to 
the random variable $E:=\sum_{j=1}^k \frac{E_j}{a_{E,j}}
=k A + \sum_{j=1}^k 
\frac{b_{E,j}}{a_{E,j}} X_{2,j}
$
and its orthogonal complements $\hat{E}_j$ $(j=1, \ldots,k-1)$.
The orthogonal complements $\hat{E}_j$ are orthogonal to $B$ as well as to $E$.
Thus, all of Eve's information for $B$ are converted to 
the pair of $E$ and $Y$.
Therefore, we can apply Theorem \ref{T2X}.
Notice that
$\sum_{j=1}^k \frac{b_{E,j}}{a_{E,j}} X_{2,j}
$ is a random variable whose variance is
$\sum_{j=1}^k (\frac{b_{E,j}}{a_{E,j}})^2$.

When all of $a_{E,j}$ and all of $b_{E,j}$ are the same values $a_E$ and $b_E$,
$E$ can be written as $k A+ \frac{b_E}{a_E}\sqrt{k} X_2$
by using another standard Gaussian variable $X_2$.
That is, $E$ has the same information as 
$a_E A+ \frac{b_E}{ \sqrt{k}} X_2$. 
We can apply the above analysis with replacement of $b_E$ by $\frac{b_E}{\sqrt{k}}$.
Hence, even when there is a possibility that Eve prepares plural antennas,
when Alice and Bob set the constant $b_E$ to be a sufficiently small number,
they can prevent the multi-antenna attack.
When Eve prepares infinitely many antennas, Alice and Bob cannot disable Eve to access their secret information.
However, considering the constraint for Eve's budget, Alice and Bob can assume a reasonable value for the constant $b_E$.

\subsection{Complex number case}\Label{s24}
In the real wireless communication, 
all of number numbers are given as complex numbers.
In this case, the random variables $A$, $X_j$, and $Y$ are given as 
$A_{R}+ i A_{I}$, $X_{j,R}+ i X_{j,I}$, and $Y_{R}+ i Y_{I}$,
where the random variables 
$A_{R}$, $A_{I}$, $X_{j,R}$ and $ X_{j,I}$ are 
independently subject to the standard Gaussian distribution.
and $Y_{R}$ and $ Y_{I}$ are independently of other random variables.
So, the noise injecting attack model with Assumptions (A1)-(A6) can be written as
\begin{align}
B& =a_{B} e^{i \theta_{B}} A+ e^{i \theta_{Y}} Y +b_{B} e^{i \theta_{1}} X_1 +e_{B} e^{i \theta_{3}}\\
E& =a_{E} e^{i \theta_{E}} A +b_{E} e^{i \theta_{2}} X_2 ,
\end{align}
where $a_B,a_E,b_B$, and $b_E$ are positive real numbers.
Now, we choose 
$B_{R}$ and $B_{I}$ ($E_{R}$ and $E_{I}$) to be the real and imaginary parts of
$e^{-i \theta_{B}}  B$ ($e^{-i \theta_{E}} E$), respectively.
We introduce 
the new random numbers
$X_{1,R}'$ and $X_{1,I}' $ as the real and imaginary parts of 
$e^{i (\theta_{1}-\theta_{B}) } X_1$, 
and 
 $X_{2,R}'$ and $X_{2,I}' $ as those of $e^{i (\theta_{2}- \theta_{E})} X_2$.
In the same way, 
we introduce $Y_{R}'$ and $Y_{I}' $ as the real and imaginary parts of 
$e^{i (\theta_{Y}-\theta_{B}) } Y$.

Then, these random numbers are also independently subject to the standard Gaussian distribution.
Thus, the noise injecting attack model with real random variables can be applied as
\begin{align}
B_R& =a_{B} A_{R}+b_{B} Y_{R}' +b_{B} X_{1,R}'
+e_{B} \cos (\theta_{4}-\theta_{E})
\\
B_I& =a_{B} A_{I}+b_{B} Y_{I}' +b_{B} X_{1,I}' 
+e_{B} \sin (\theta_{4}-\theta_{E})
\\
E_R& =a_{E} A_{R} +b_{E} X_{2,R}' 
\\
E_I& =a_{E} A_{I} +b_{E} X_{2,I}'.
\end{align}

\section{Interference model}\Label{s21}
In this section, 
to understand the model of this paper, 
we discuss another model, in which, Eve is weaker than 
Eve of the present model.
we consider a different situation as a preparation for our analysis of
noise injecting attack.
That is, we replace Assumption (A1) by the following assumption (A1)',
and keep the quasi static assumption (Assumption (A4)) for both channels.
\begin{description}
\item[(A1)']
Bob's and Eve's detections $B$ and $\tilde{E}$ are written 
by using Alice's signal $A$ as
\begin{align}
B&:= {a_B} A+ Y+ {b_B} X_1+ e_B, \Label{4-24-1} \\
\tilde{E}&:= {a_E} A+ Y_2+ {b_E} X_2, \Label{4-24-2}
\end{align}
where the random variables $X_j$ are subject to the standard Gaussian distribution independently 
and $a_B,b_B$, $e_B$, $a_E,b_E$, $b_E$ are constants during the coherent time, i.e., they can be regarded for one block length for our code.
The remaining random variables $Y$ and $Y_2$ are independent of $X_j$.
\end{description}

Due to the quasi static assumption (Assumption (A4)) for both channels,
the variables $X_1,X_2,Y$, and $Y_2$ are independent in each transmission.
Eve is not allowed to know the value of $Y$ and $Y_2$,
but she knows their distribution.
Hence, Assumption (A6) is replaced by the following.
\begin{description}
\item[(A6)']
Eve knows 
all the channel parameters, which are given in \eqref{4-24-1} and \eqref{4-24-2}.
Also, she knows the forms of the distributions of $Y$ and $Y_2$.
\end{description}

Since there is no assumption for the relation between $Y$ and $Y_2$, this model contains the case with interference.
Under this model, the second terms express the noise during transmission over the space between Alice and Bob,
and the third terms express thermal noises in the individual detectors
due to the local Gaussian noise assumption (Assumption (A2)), whose 
sizes are reflected in the constants $b_E$, and $b_E$.

As discussed in Section \ref{s0},
when Eve is closer to Alice than Bob and 
the performance of Eve's detector is the same as that of Bob's,
secure communication with forward reconciliation is impossible.
Since there is a possibility of secure communication with backward reconciliation,
this section also discusses this type of secure key generation.
Hence, we consider the same protocol with backward reconciliation, as given in Section \ref{s3}. 

Under Assumptions (A1)', (A2)-(A5), (A6)', Alice and Bob can estimate the coefficients $a_B$, $b_B$, $e_B$, $a_E$, and $b_E$ in the same way as in the model discussed above.
Also, they can estimate the distribution of $Y$.
However, it is impossible for them to estimate the distribution of $Y_2$, i.e., type of interference during transmission.
To evaluate the security of this case,
we focus on the Markov chain 
\begin{align}
B\mc (Y,E)\mc (Y_2,E)\mc \tilde{E}.\Label{MFT}
\end{align}
Hence, Eve of this model is not stronger than Eve in Section \ref{s01}.
Therefore, it is sufficient to evaluate the information leakage to Eve in Section \ref{s01}.
That is, the security evaluation in Section \ref{s33} is sufficient for the model of this section.

In fact, when $Y_2= \frac{a_E^2+b_E^2}{a_B a_E} Y$,
we have $E'=\frac{a_B a_E}{a_E^2+b_E^2}\tilde{E} $.
Due to Theorem \ref{T2X},
Eve's performance of this model equals that of the model in Section \ref{s01}. 
Therefore, Eve of this model is sufficiently strong  under this special interference.
This kind of characterization is based on the local Gaussian noise assumption.

\begin{figure}[htbp]
\begin{center}
\scalebox{1}{\includegraphics[scale=0.38]{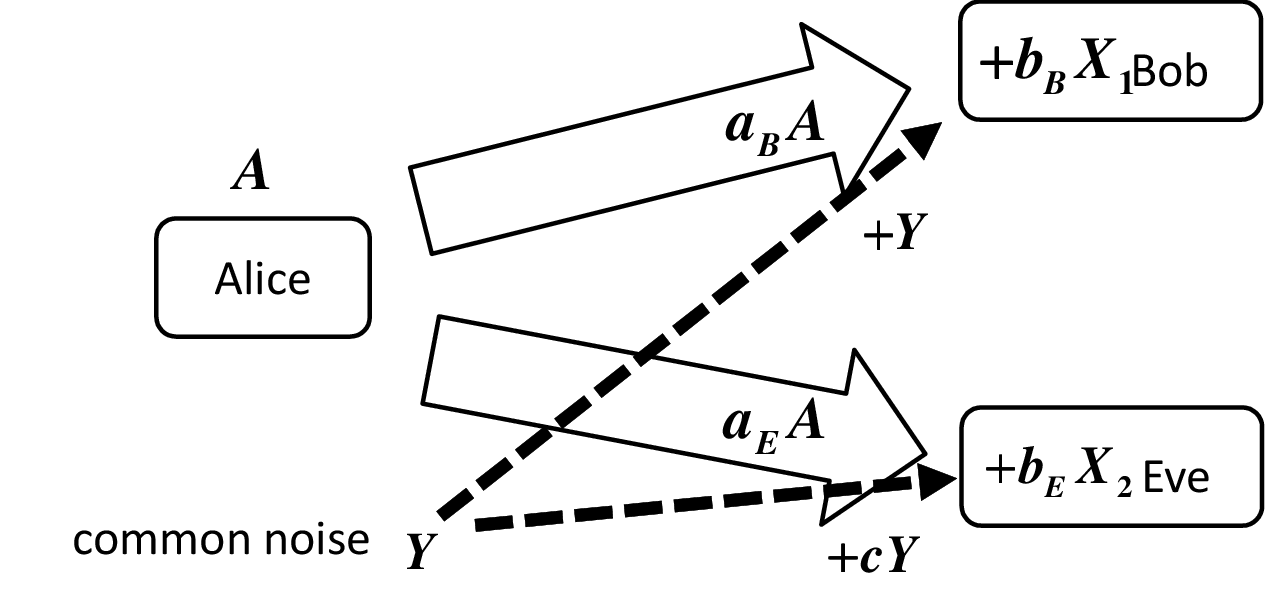}}
\end{center}
\caption{Interference model:
This figure shows the case when 
Eve has the same interference as that in Bob's detection.}
\Label{F1}
\end{figure}%

\section{Proofs}\Label{s5}
\subsection{Proofs of 
Lemmas \ref{L30} and \ref{LLO}}\Label{s53}
We introduce Gallager function
\begin{align*}
E_0(t|P_{B'|X},P_{B'}) 
\!:=\!
\log \!\int_{-\infty}^\infty 
\!(\sum_{b'} \!P_{B'}(b)\! P_{X|B'}(x|b')^{\frac{1}{1\!-\!t}})^{1\!-\!t} dx,
\end{align*}
which is known to be convex for $t$ \cite{Gal}.
Since $2^{\phi[P_{E'},v](t)}
=2^{E_0(s|P_{B'|X},P_{B'})}((\frac{1}{2})^{1-\frac{1}{1-t}})^{1-t}$,
we have 
$\phi[P_{E'},v](t)=E_0(t|P_{B'|X},P_{B'})+t$, which 
shows that $\phi[P_{E'},v](t)$ is convex for $t$.
Hence, we obtain Lemma \ref{L30}.

Now, we show Lemma \ref{LLO}, i.e., \eqref{28-1} and \eqref{28-2}.
For this purpose, we introduce 
a function for a joint distribution
$P_{X,Y}$ as
$\phi(t|X|Y|P_{X,Y}):=
\int_{{\cal Y}} (\sum_{x} P_{X,Y}(x,y)^{\frac{1}{1-t}})^{1-t} d y$, which is denoted by 
$-tH_{\frac{1}{1-t}}^{\uparrow}(X|Y|P_{X,Y})$ 
in \cite{H-ep} or
$-tH_{\frac{1}{1-t}}^{\rm G}(X|Y|P_{X,Y})$ in \cite{H-cq}.
In this proof, we employ the rates $R_1:= \frac{\dim C}{n}$
and $R_2:= \frac{m_1}{n}$.
The function $\phi(t|X|Y|P_{X,Y})$ is a generalization of $\phi[P_{E'},v](t)$.
Applying \cite[(67)]{H-tight} and \cite[(21)]{H-cq}, we have
\begin{align}
&d(K:{\cal E}|H) \nonumber \\
 \stackrel{(a)}{\le}&
3 \min_{t \in [0,\frac{1}{2}]} 
2^{t n (R_1-R_2)+\phi(t|{B'}^n|[{B'}^n],E^n|P_{{B'}^n,E^n}) } \nonumber \\
 \stackrel{(b)}{\le} &
3 \min_{t \in [0,\frac{1}{2}]} 
2^{t n (R_1-R_2)+t n(1 -R_1)}
2^{\phi(t|{B'}^n,[{B'}^n]|E^n|P_{{B'}^n,E^n})} \nonumber \\
=&
3 \min_{t \in [0,\frac{1}{2}]} 
2^{t n(1 -R_2)}2^{n \phi(t|B'|E|P_{B',E})}\nonumber  \\
=&
3 \min_{t \in [0,\frac{1}{2}]} 
2^{t n(1 -R_2)}2^{n \phi[P_{E'},v](t)},
\end{align}
where $(a)$ and $(b)$ follow from \cite[(67)]{H-tight} and \cite[(21)]{H-cq}, respectively.
So, we obtain \eqref{28-2}.
When we replace the role of \cite[(67)]{H-tight} by 
\cite[(54) and Lemma 22]{H-ep}, we obtain a similar evaluation as \eqref{28-2} for 
$\epsilon$-almost universal dual hash function.

Now, we introduce another function for 
$P_{X,Y}$ as
\begin{align*}
H_{1+s}(X|Y|P_{X,Y})\!:=\!
-\frac{1}{s}\log \!
\int_{{\cal Y}} \!(\sum_{x} P_{X|Y}(x|y)^{1+s}) P_Y(y)dy.
\end{align*}
We denote $s H_{1+s}(X|Y|P_{X,Y})$ by $\tilde{H}_{1+s}(X|Y|P_{X,Y})$ in \cite{H-leaked} or $s{H}^{\downarrow}_{1+s}(X|Y|P_{X,Y})$ in \cite{H-ep}.
Applying \cite[(3)]{H-leaked}, \cite[Lemma 5]{H-ep}, and \cite[(21)]{H-cq}, we have
\begin{align}
&I'(K:{\cal E}|H) \nonumber  \\
 \stackrel{(a)}{\le} &
\inf_{s \in (0,1) }
\frac{1}{s}
2^{s n (R_1-R_2)-s H_{1+s}({B'}^n|[{B'}^n],E^n|P_{{B'}^n,E^n}) }\nonumber  \\
 \stackrel{(b)}{\le} &
\inf_{s \in (0,1) }
\frac{1}{s}
2^{s n (R_1-R_2)+s n(1 -R_1)} 2^{\phi(s|{B'}^n,[{B'}^n]|E^n|P_{{B'}^n,E^n})} \nonumber \\
 \stackrel{(c)}{\le} &
\inf_{s \in (0,1) }
\frac{1}{s}
2^{s n (R_1-R_2)+s n(1 -R_1)}2^{\phi(s|{B'}^n,[{B'}^n]|E^n|P_{{B'}^n,E^n})} \nonumber \\
=&
\inf_{s \in (0,1) }
\frac{1}{s}
2^{s n(1 -R_2)} 2^{n \phi(s|B'|E|P_{B',E})} \nonumber \\
=&
\inf_{s \in (0,1) }
\frac{1}{s}
2^{s n(1 -R_2)} 2^{n \phi[P_{E'},v](s)},
\end{align}
where 
$(a)$, $(b)$, and $(c)$ follow from \cite[(3)]{H-leaked}, \cite[Lemma 5]{H-ep}, and \cite[(21)]{H-cq}, respectively.
So, we obtain \eqref{28-1}.
When we need an evaluation with $\epsilon$-almost universal dual hash function,
it is sufficient to replace the role of \cite[(3)]{H-leaked} by \cite[(56) and Theorem 23]{H-ep}.

\subsection{Proof of Lemma \ref{L11-10}}\Label{s11-17-1}

Next, we prepare the following two lemmas
\begin{lem}\Label{L11-19-2}
When a distribution function $F_1$ is differentiable,
three distribution functions $F_1$, $F_2$, and $F_3$ satisfy
\begin{align}
\sup_x|F_1 * F_2 (x) - F_1 * F_3 (x)|
\le
\sup_x|F_2 (x) - F_3 (x)|.
\end{align}
\end{lem}
\begin{proof}
Since $ \int_{-\infty}^\infty  \frac{d F_1}{dx}(x-y)  dy
= \int_{-\infty}^\infty  \frac{d F_1}{dx}(y'-x)  dy'=1$ with $y'=y$, we have
\begin{align*}
& \sup_x|F_1 * F_2 (x) - F_1 * F_3 (x)|
=
\sup_x|  \int_{-\infty}^\infty \frac{d F_1}{dx}(x-y) (F_2 (y) - F_3 (y) dy| \\
\le &
\sup_x   \int_{-\infty}^\infty \frac{d F_1}{dx}(x-y) (\sup_{y'}| F_2 (y') - F_3 (y')|) dy
=\sup_x \int_{-\infty}^\infty   \frac{d F_1}{dx}(x-y)  dy (\sup_{y'}| F_2 (y') - F_3 (y')|) \\
=& \sup_{y'}| F_2 (y') - F_3 (y')|.
\end{align*}
\end{proof}

\begin{lem}\Label{L11-19-1}
When $a<1 $,
\begin{align}
\sup_x|\Phi_{1} (x) - \Phi_{a}(x)|
=
\int^{\sqrt{\frac{2 \log a}{1-a^2}}}_{a \sqrt{\frac{2 \log a}{1-a^2}} }
\frac{1}{\sqrt{2\pi}}
e^{-\frac{x^2}{2}}d x.\Label{11-19-1}
\end{align}
When $a \ge 1 $,
\begin{align}
\sup_x|\Phi_{1} (x) - \Phi_{a}(x)|
=
\int^{a \sqrt{\frac{2 \log a}{1-a^2}} }_{\sqrt{\frac{2 \log a}{1-a^2}}}
\frac{1}{\sqrt{2\pi}}
e^{-\frac{x^2}{2}}d x.\Label{11-19-2}
\end{align}
When $a$ is close to $1$, 
\begin{align}
\sup_x|\Phi_{1} (x) - \Phi_{a}(x)|
=
\frac{|a-1| }{\sqrt{2\pi e}}+ O((a-1)^2).\Label{11-19-3}
\end{align}
\end{lem}

\begin{proof}
The derivative of $\Phi_{1} (x) - \Phi_{a}(x)$ with respect to $x$ is
$\frac{1}{\sqrt{2\pi}}
(e^{-\frac{x^2}{2}}
- \frac{1}{a}e^{-\frac{x^2}{2a^2}})$.
It equals zero if and only if
$x= \pm a \sqrt{\frac{-2\log a}{1-a^2}}$.
Due to the symmetry $\Phi_{1} (x) - \Phi_{a}(x)=-\Phi_{1} (-x) - \Phi_{a}(-x)$,
the maximum of the absolute value $|\Phi_{1} (x) - \Phi_{a}(x)|$
is realized when $x= \pm a \sqrt{\frac{-2\log a}{1-a^2}}$.
Substituting this value, we obtain \eqref{11-19-1} and \eqref{11-19-2}.

When $a$ is close to $1$,
$\frac{-2\log a}{1-a^2}
= \frac{-2\log (1+a-1)}{1-a^2}
= \frac{-2((a-1)+ O((a-1)^2))}{1-a^2}
=\frac{2}{a+1}+O(a-1)
=1+O(a-1)$.
So, we have
$e^{-\frac{\frac{-2\log a}{1-a^2}}{2}}
=e^{-\frac{1}{2}}+O(a-1)$,
$e^{-\frac{\frac{-2a\log a}{1-a^2}}{2}}
=e^{-\frac{1}{2}}+O(a-1)$,
and
$\sqrt{\frac{-2\log a}{1-a^2}}-a \sqrt{\frac{-2\log a}{1-a^2}}
=(1-a)+O((a-1)^2)$.
Thus,
\begin{align}
\int^{\sqrt{\frac{2 \log a}{1-a^2}}}_{a \sqrt{\frac{2 \log a}{1-a^2}} }
\frac{1}{\sqrt{2\pi}}
e^{-\frac{x^2}{2}}d x
=|(a-1)+O((a-1)^2)|\cdot
\frac{1 }{\sqrt{2\pi}}(e^{-\frac{1}{2}}+O(a-1))
=\frac{|a-1| }{\sqrt{2\pi e}}+ O((a-1)^2),
\end{align}
which implies \eqref{11-19-3}.
\end{proof}

Now, we show \eqref{e11-10} under the condition $\frac{\hat{c}_{AB}^2 a_E^2}{a_E^2+b_E^2}\le b_B^2$.
We notice that
\begin{align}
& \sup_x
\Big|{F}_{E'}(x)
-\hat{F}_{E'}(x)\Big| \nonumber \\
\le &
\sup_x
\Big|{F}_{E'}(x)
-\Phi_{\sqrt{ \frac{\hat{c}_{AB}^2 a_E^2}{a_E^2+b_E^2}- b_B^2}} * F_{\bar{A}^c}(x)\Big| \nonumber \\
&+
\sup_x
\Big|\Phi_{\sqrt{ \frac{\hat{c}_{AB}^2 a_E^2}{a_E^2+b_E^2}- b_B^2}} * F_{\bar{A}^c}(x)
-\hat{F}_{E'}(x)\Big| .
\end{align}

We discuss the first term.
Since
\begin{align}
& \Phi_{\sqrt{ \frac{\hat{c}_{AB}^2 a_E^2}{a_E^2+b_E^2}- b_B^2}} *{F}_{\bar{A}^c}(x)
=
\Phi_{\sqrt{ \frac{\hat{c}_{AB}^2 a_E^2}{a_E^2+b_E^2}- b_B^2}} *
F_{B-\hat{c}_{AB}A -\hat{e}_B }(x) \nonumber \\
=&
\Phi_{\sqrt{ \frac{\hat{c}_{AB}^2 a_E^2}{a_E^2+b_E^2}- b_B^2}} *
F_{Y+b_B X_1 +a_B A +e_B -\hat{c}_{AB}A -\hat{e}_B }(x)
=
\Phi_{\sqrt{ \frac{\hat{c}_{AB}^2 a_E^2}{a_E^2+b_E^2}- b_B^2}} *
F_{b_B X_1 +(a_B -\hat{c}_{AB})A +Y+e_B -\hat{e}_B }(x) \nonumber \\
=&
\Phi_{\sqrt{ \frac{\hat{c}_{AB}^2 a_E^2}{a_E^2+b_E^2}- b_B^2}} *
(\Phi_{\sqrt{ b_B^2+ (a_B -\hat{c}_{AB})^2}}* 
F_{Y+e_B -\hat{e}_B })(x) 
=
\Phi_{\sqrt{ \frac{\hat{c}_{AB}^2 a_E^2}{a_E^2+b_E^2} + (a_B -\hat{c}_{AB})^2}} *
F_{Y+e_B -\hat{e}_B }(x),
\end{align}
we have
\begin{align}
& \sup_x
\Big|{F}_{E'}(x)
-\Phi_{\sqrt{ \frac{\hat{c}_{AB}^2 a_E^2}{a_E^2+b_E^2}- b_B^2}} *{F}_{\bar{A}^c}(x)\Big| \nonumber \\
=&
\sup_x
\Big|\Phi_{\frac{{a}_{B}  a_E}{\sqrt{a_E^2+b_E^2}}} *{F}_{Y+e_B -\hat{e}_B}(x)
-
\Phi_{\sqrt{ \frac{\hat{c}_{AB}^2 a_E^2}{a_E^2+b_E^2}+ (a_B -\hat{c}_{AB})^2}} *{F}_{Y+e_B -\hat{e}_B}(x)\Big|
\nonumber \\
 \stackrel{(a)}{\le}  &
\sup_x
\Big|\Phi_{\frac{{a}_{B}  a_E}{\sqrt{a_E^2+b_E^2}}} (x)
-
\Phi_{\sqrt{ \frac{\hat{c}_{AB}^2 a_E^2}{a_E^2+b_E^2}+ (a_B -\hat{c}_{AB})^2}} (x)\Big|
\nonumber \\
\stackrel{(b)}{=} &
\frac{1}{\sqrt{2\pi e}}
\Big|\frac{\hat{c}_{AB} }{a_B}-1\Big|^2
+O((\frac{\hat{c}_{AB} }{a_B}-1)^2),
\end{align}
where $(a)$ follows from Lemma \ref{L11-19-2} and $(b)$ does from the combination of \eqref{11-19-3} and the following derivation;
\begin{align}
& \frac{\sqrt{ \frac{\hat{c}_{AB}^2 a_E^2}{a_E^2+b_E^2}+ (a_B -\hat{c}_{AB})^2}}
{\frac{{a}_{B}  a_E}{\sqrt{a_E^2+b_E^2}}} \nonumber \\
=&
\sqrt{ \frac{\hat{c}_{AB}^2 }{a_B^2}
+ \frac{a_E^2+b_E^2}{ a_E^2}(1-\frac{\hat{c}_{AB}}{a_B})^2}
\nonumber \\
=&
\frac{\hat{c}_{AB} }{a_B}\sqrt{ 1
+ \frac{a_E^2+b_E^2}{ a_E^2}(\frac{a_B}{\hat{c}_{AB}}-1)^2} \nonumber \\
=&
\frac{\hat{c}_{AB} }{a_B}
( 1
+ \frac{a_E^2+b_E^2}{ 2 a_E^2}(\frac{a_B}{\hat{c}_{AB}}-1)^2
+O((\frac{\hat{c}_{AB}}{a_B}-1)^4))
\nonumber  \\
=&
\frac{\hat{c}_{AB} }{a_B}+ \frac{a_E^2+b_E^2}{ 2 a_E^2} \frac{\hat{c}_{AB} }{a_B}
(\frac{a_B}{\hat{c}_{AB}}-1)^2
+O((\frac{\hat{c}_{AB}}{a_B}-1)^4)
\nonumber  \\
=&
\frac{\hat{c}_{AB} }{a_B}
+O((\frac{\hat{c}_{AB}}{a_B}-1)^2).
\end{align}

When $l\ge l_{\epsilon,\delta}$,
with confidence level $1-(\epsilon+\delta)$, 
since the relation 
$|\frac{\hat{c}_{AB} }{a_B}-1|\le \frac{\sqrt{\hat{v}_{AB}}}{\hat{c}_{AB}
\sqrt{l}}Z_\epsilon$
holds, 
the first term is upper bounded by 
$ \frac{1}{\sqrt{2\pi e}}
\cdot \frac{\sqrt{\hat{v}_{AB}}}{\hat{c}_{AB}
\sqrt{l}}Z_\epsilon$.

The second term is evaluated as
\begin{align}
& \sup_x
\Big|\Phi_{\sqrt{ \frac{\hat{c}_{AB}^2 a_E^2}{a_E^2+b_E^2}- b_B^2}} * F_{\bar{A}^c}(x)
-\hat{F}_{E'}(x) \Big| \nonumber \\
=&
\sup_x
\Big|\Phi_{\sqrt{ \frac{\hat{c}_{AB}^2 a_E^2}{a_E^2+b_E^2}- b_B^2}} * F_{\bar{A}^c}(x)
-\Phi_{\sqrt{ \frac{\hat{c}_{AB}^2 a_E^2}{a_E^2+b_E^2}- b_B^2}} * \hat{F}_{\bar{A}^c}(x)
\Big|\nonumber  \\
 \stackrel{(a)}{\le}  &
\sup_x
|F_{\bar{A}^c}(x)-\hat{F}_{\bar{A}^c}(x)| ,
\end{align}
where $(a)$ follows from Lemma \ref{L11-19-2}.
So, 
when $l\ge l_{\epsilon,\delta}$,
the second term is upper bounded by 
$ \frac{1}{\sqrt{l}}L^{-1}(1-\epsilon)$
with confidence level $1-(\epsilon+\delta)$.

Combining these two discussion, we obtain 
\eqref{e11-10} with confidence level $1-2(\epsilon+\delta)$.
Since the relation 
$|\frac{\hat{c}_{AB} }{a_B}-1|\le \frac{\sqrt{\hat{v}_{AB}}}{\hat{c}_{AB}
\sqrt{l}}Z_\epsilon$
holds with confidence level $1-(\epsilon+\delta)$ 
when $l\ge l_{\epsilon,\delta}$,
combining \eqref{eq11-19-1}, we obtain \eqref{e11-10}.

Now, we show \eqref{e11-10-2}.
We notice that
\begin{align}
& \sup_x
|{F}_{E''}(x)
-\hat{F}_{E''}(x)| \nonumber \\
\le &
\sup_x
|{F}_{E''}(x)
-{F}_{\bar{A}^c}(x)|\nonumber  \\
&+
\sup_x
|{F}_{\bar{A}^c}(x)
-\hat{F}_{\bar{A}^c}(x)| .
\end{align}
Since
\begin{align}
{F}_{E''}
={F}_{E''+e_B-\hat{e}_B}
={F}_{A^c+e_B-\hat{e}_B}
={F}_{B-a_B A-e_B +e_B-\hat{e}_B}
={F}_{B-a_B A-\hat{e}_B}
=\Phi_{a_B}*{F}_{B-\hat{e}_B}
\end{align}
and
\begin{align}
{F}_{\bar{A}^c}
={F}_{B-\hat{c}_{AB} A-\hat{e}_B}
=\Phi_{\hat{c}_{AB}}*{F}_{B-\hat{e}_B},
\end{align}
the first term is evaluated as
\begin{align}
&
\sup_x
|{F}_{E''}(x)
-{F}_{\bar{A}^c}(x)| 
=
\sup_x
|\Phi_{a_B}*{F}_{B-\hat{e}_B}(x)
- \Phi_{\hat{c}_{AB}}*{F}_{B-\hat{e}_B}(x) | \nonumber \\
\stackrel{(a)}{\le} &
\sup_x
|\Phi_{a_B}(x)- \Phi_{\hat{c}_{AB}}(x) | \nonumber \\
\stackrel{(b)}{=} &
\frac{1}{\sqrt{2\pi e}}
\Big|\frac{\hat{c}_{AB} }{a_B}-1\Big|^2
+O((\frac{\hat{c}_{AB} }{a_B}-1)^2),
\end{align}
where $(a)$ and $(b)$ follow from Lemma \ref{L11-19-2} and \eqref{11-19-3},
respectively.
Thus, in the same way, we obtain \eqref{e11-10-2}.

\subsection{Proof of Lemma \ref{L11-11b}}\Label{s11-17-2}
To show Lemma \ref{L11-11b}, we prepare the following lemmas.

\begin{lem}\Label{L11-22}
The function $v \mapsto 
(\Phi_1(\frac{x}{\sqrt{v}})^{\frac{1}{1-t}}+(1-\Phi_1(\frac{x}{\sqrt{v}}))^{\frac{1}{1-t}}
)^{1-t}$ is monotone decreasing for any $x$.
\end{lem}

\begin{proof}
The function $p \mapsto p^{^{\frac{1}{1-t}}}+(1-p)^{\frac{1}{1-t}}$
is monotone decreasing in $[0,\frac{1}{2}]$
and is monotone increasing in $[\frac{1}{2},1]$.
The function $v \mapsto \Phi_1(\frac{x}{\sqrt{v}})$ 
is monotone increasing in $[0,1]$ for $x>0$
and is monotone decreasing in $[0,1]$ for $x<0$.
Since $\Phi_1(\frac{1}{\sqrt{v}} 0)=\frac{1}{2}$, 
we conclude that the function $v \mapsto 
(\Phi_1(\frac{x}{\sqrt{v}})^{\frac{1}{1-t}}+(1-\Phi_1(\frac{x}{\sqrt{v}}))^{\frac{1}{1-t}}
)$ is monotone increasing for any $x$.
\end{proof}

\begin{lem}\Label{L11-11}
When a non-negative valued function $f $ is monotone decreasing in $(-\infty,0]$ and 
is monotone increasing in $[0,\infty)$,
we have
\begin{align}
& |\int_{-\infty}^{\infty} f(x) P_X(dx) - \int_{-\infty}^{\infty} f(x) P_{X'}(dx)| \nonumber \\
\le &
2 \sup_x (f(x)-f(0))\sup_{x}|F_X(x)-F_{X'}(x)|.\Label{11-22-7}
\end{align}
\end{lem}
Lemma \ref{L11-11} will be shown in Subsection \ref{s11-17-3}.

When $l$ is sufficiently large,
Lemma \ref{L11-10} and its proof guarantee 
\eqref{e11-10} and 
$|\frac{\hat{c}_{AB}}{a_B}-1|
\le \frac{\sqrt{\hat{v}_{AB}}}{\hat{c}_{AB} \sqrt{l}}Z_{\epsilon}$  
with confidence level almost $1-2 \epsilon$.
In the following discussion, we give a statement 
with confidence level almost $1-2 \epsilon$.
Thus, Lemma \ref{L11-22} guarantees that
\begin{align}
2^{\phi[P_{E'},v_{B|E'}](t)}
\le 
2^{\phi[P_{E'},\underline{v}_{B|E'}]}(t)
\end{align}

Now, we apply Lemma \ref{L11-11} to the case when 
$f(x)=(\Phi (\frac{x}{\underline{v}_{B|E'}})^{\frac{1}{1-t}}
+(1-\Phi (\frac{x}{\underline{v}_{B|E'}}))^{\frac{1}{1-t}})^{1-t}$,
which satisfies the condition of Lemma \ref{L11-11}.
The supremum $\sup_x (f(x)-f(0))$ is calculated as
\begin{align}
\sup_x (f(x)-f(0))
= 1-\Big( \frac{2}{2^{\frac{1}{1-t}}}\Big)^{1-t} =1- 2^{-t}.\Label{11-22-1}
\end{align}
Thus, we have
\begin{align}
& |
2^{\phi[{P}_{E'},\underline{v}_{B|E'}]}(t)
-2^{\phi[\hat{P}_{E'},\underline{v}_{B|E'}]}(t)
| \nonumber \\
\le &
2\sup_x (f(x)-f(0))\sup_{x}|F_{E'}(x)-\hat{F}_{E'}(x)|
=
2(1- 2^{-t})
\sup_x 
|F_{E'}(x)-\hat{F}_{E'}(x)| \nonumber \\
\le &
2(1- 2^{-t})
\Big(\frac{\sqrt{\hat{v}_{AB}}}{\sqrt{2\pi e}\hat{c}_{AB} \sqrt{l}}Z_{\epsilon}
+\frac{1}{\sqrt{l}} L^{-1}( 1-\epsilon) \Big).\Label{11-22-2}
\end{align}
So, combining \eqref{11-22-1} and \eqref{11-22-2},
we obtain \eqref{11-20-3} when 
$\frac{\hat{c}_{AB}^2 a_E^2}{a_E^2+b_E^2}\ge b_B^2$.

Using the same discussion of the proof of the above case
and the Markovian chain $B \mc E'' \mc E'$ shown as \eqref{E16D}, 
we can show \eqref{11-20-3} 
when $\frac{\hat{c}_{AB}^2 a_E^2}{a_E^2+b_E^2}< b_B^2$.

\subsection{Proof of Lemma \ref{L11-11}}\Label{s11-17-3}
To show Lemma \ref{L11-11}, we prepare the following lemma.
\begin{lem}\Label{L11-11T}
Let $f$ be a continuous function defined on $[a,b]$,
and $G$ be a bounded function defined on $[a,b]$.
We assume that 
$G$ is differentiable except for a finite number of discontinuous points.
\begin{description}
\item[(D1)]
When the function $f $ is monotone decreasing in $[a,b]$
and the function $G$ satisfies
\begin{align}
\max_{y \in [a,b]}(G(y)-G(a)) =G(b)-G(a),
\end{align}
we have
\begin{align}
\int_a^b f(x) \frac{d G}{dx}(x) dx
\le f(a) \int_a^b \frac{d G}{dx}(x) dx.
\end{align}

\item[(D2)]
When the function $f $ is monotone increasing in $[a,b]$
and the function $G$ satisfies
\begin{align}
\min_{y \in [a,b]}(G(y)-G(a)) =G(b)-G(a),
\end{align}
we have
\begin{align}
\int_a^b f(x) \frac{d G}{dx}(x) dx
\le f(a) \int_a^b \frac{d G}{dx}(x) dx.
\end{align}

\item[(D3)]
When the function $f $ is monotone increasing in $[a,b]$
and the function $G$ satisfies
\begin{align}
\max_{y \in [a,b]}
(G(b)-G(y)) =G(b)-G(a),
\end{align}
we have
\begin{align}
\int_a^b f(x) \frac{d G}{dx}(x) dx
\le f(b) \int_a^b \frac{d G}{dx}(x) dx.
\end{align}

\item[(D4)]
When the function $f $ is monotone decreasing in $[a,b]$
and the function $G$ satisfies
\begin{align}
\min_{y \in [a,b]}(G(b)-G(y)) =G(b)-G(a),
\end{align}
we have
\begin{align}
\int_a^b f(x) \frac{d G}{dx}(x) dx
\le f(b) \int_a^b \frac{d G}{dx}(x) dx.
\end{align}
\end{description}
\end{lem}

\begin{proofof}{Lemma \ref{L11-11T}}
We first show Item (D1).
We assume that $G$ is $C^1$-continuous.
We assume that there are points 
$a=a_1, a_2, \ldots, a_{2n}=b$ such that
$g(x)\le 0$ for $a_{2i}<x<a_{2i+1} $
and
$g(x)\ge 0$ for $a_{2i+1}<x<a_{2i+2} $.
So, the assumption of this lemma implies that
$\int_{a_{2n-2}}^{a_{2n-1}} g(x) dx
+ \int_{a_{2n-1}}^{a_{2n}} g(x) dx \ge 0$.
We choose a point $a_{2n+1}'$ such that
$\int_{a_{2n-2}}^{a_{2n-1}} g(x) dx
+ \int_{a_{2n-1}}^{a_{2n-1}'} g(x) dx = 0$.
We define the function $g_1$ as
\begin{align}
g_1(x):=
\left\{
\begin{array}{ll}
g(x) & \hbox{ for }x< a_{2n-2}\\
0 & \hbox{ for } a_{2n-2} \le x \le a_{2n-1}' \\
g(x) & \hbox{ for } a_{2n-1}' \le x  .
\end{array}
\right.
\end{align}
Then, we have
\begin{align}
\int_a^b f(x) g(x) dx
\le 
\int_a^b f(x) g_1(x) dx.
\end{align}
and
$g_1(x) \ge 0$ for $a_{2n-3} < x < b  $.
Rewriting $a_{2n}$ by $a_{2n-2}$, we repeat the above process for $g_1$
and denote the resultant function by $g_2$.
Repeating this procedure, we define $g_1,g_2,\ldots, g_n$.
So, $g_n$ satisfies $g_n(x)\ge 0$ on $(a,b)$ and
\begin{align}
\int_a^b f(x) g(x) dx
\le 
\int_a^b f(x) g_n(x) dx.
\end{align}
Since 
\begin{align}
\int_a^b f(x) g_n(x) dx
\le 
f(a)\int_a^b g_n(x) dx,
\end{align}
we obtain the desired statement of Item (D1)
when $G$ is $C^1$-continuous.
In the general case,
$G$ can be approximated by a $C^1$-continuous function
satisfying the desired conditions.
So, we obtain the desired statement of Item (D1) in the general case.

Applying $-f$ and $-g$ to Item (D1), we obtain Item (D2).
Applying $f(a+b-x)$ and $g(a+b-x)$ to Items (D1) and (D2), we obtain Items (D3) and (D4), respectively.
\end{proofof}

\begin{proofof}{Lemma \ref{L11-11}}
To show \eqref{11-22-7}, it is enough to discuss the case when $f(0)=0$
because the general case can be obtained by substituting $f(x)-f(0)$ into $f$.
Also, it is enough to show that
\begin{align}
& \int_{-\infty}^{\infty} f(x) P_X(dx) - \int_{-\infty}^{\infty} f(x) P_{X'}(dx) \nonumber \\
\le &
2 \sup_x (f(x)-f(0))\sup_{x}|F_X(x)-F_{X'}(x)|.\Label{11-22-7B}
\end{align}
We choose
\begin{align}
x_1&:= \argmax_{x \in (-\infty,0]} (F_X(x)-F_{X'}(x)) \\
x_2&:= \argmin_{x \in [0,\infty) } (F_X(x)-F_{X'}(x)) .
\end{align}
We choose $R$ sufficiently large.

Items (D1) and (D4) of Lemma \ref{L11-11T} imply
\begin{align}
\int_{-R}^{x_1} f(x) P_X(dx) - \int_{-R}^{x_1} f(x) P_{X'}(dx)
\le &
f(-R) [( F_X(x_1)-F_X(-R) )- ( F_{X'}(x_1)-F_{X'}(-R) )] \\
\int_{x_1}^{0} f(x) P_X(dx) - \int_{x_1}^0 f(x) P_{X'}(dx)
\le &
f(0) [( F_X(0)-F_X(x_1) )- ( F_{X'}(0)-F_{X'}(x_1) )] =0,
\end{align}
respectively.
Similarly, Items (D2) and (D3) of Lemma \ref{L11-11T} imply
\begin{align}
\int_{x_2}^{R} f(x) P_X(dx) - \int_{x_2}^{R} f(x) P_{X'}(dx)
\le &
f(R) [( F_X(R)-F_X(x_2) )- ( F_{X'}(R)-F_{X'}(x_2) )] \\
\int_0^{x_2} f(x) P_X(dx) - \int_0^{x_2} f(x) P_{X'}(dx)
\le &
f(0) [( F_X(x_2)-F_X(0) )- ( F_{X'}(x_2)-F_{X'}(0) )] =0,
\end{align}
respectively.
Combining them, we have
\begin{align}
&\int_{-R}^{R} f(x) P_X(dx) - \int_{-R}^{R} f(x) P_{X'}(dx) \nonumber \\
\le &
f(-R) [( F_X(x_1)-F_X(-R) )- ( F_{X'}(x_1)-F_{X'}(-R) )] \nonumber \\
&+
f(R) [( F_X(R)-F_X(x_2) )- ( F_{X'}(R)-F_{X'}(x_2) )] .
\end{align}
Taking the limit $R \to \infty$, we have
\begin{align}
&\int_{-\infty}^{\infty} f(x) P_X(dx) - \int_{-\infty}^{\infty} f(x) P_{X'}(dx)
\nonumber \\
\le &
\lim_{R\to \infty}
f(-R) [F_X(x_1)-  F_{X'}(x_1)] 
+
\lim_{R\to \infty} f(R) [F_{X'}(x_2)-F_X(x_2) ] ,
\end{align}
which implies \eqref{11-22-7B}.
\end{proofof}


\section{Discussion}\Label{s6}
We have proposed the noise injecting attack as a very strong attack to secure wireless communication, in which,
Eve can control everything except for the neighborhood of Alice and Bob's detector. 
Under a reasonable assumption (A1)-(A6) for the performance of Eve's detector,
i.e., under the model \eqref{4-24-1X} and \eqref{4-24-2X},
we have constructed a secure key generation protocol by using backward reconciliation over the noise injecting attack.
For this analysis, as Theorem \ref{T2X},
in the noise injecting attack 
we have shown that Eve's information can be reduced to the single random variable $E'$ as
Theorem \ref{T2X}.

Also, as Lemma \ref{L1},
when the additive noise generated during the transmission is subject to a Gaussian distribution,
we have derived a necessary and sufficient condition \eqref{4-24-3} of the coefficients $a_B$ $a_E$, and $b_E$ and the variance of its Gaussian random variable
for realizing greater correlation coefficient, i.e., greater mutual information between Alice and Bob than that between Bob and Eve
under a spatial condition for Eve.
Even when it is difficult to realize the condition \eqref{4-24-3},
we have proposed the post selection method, in which, we choose only the case when the condition \eqref{4-24-3} holds by utilizing the stochastic behavior of 
the LHS of \eqref{4-24-3}.


To identify the channel between Alice and Bob,
our protocol contains the estimation of channel.
In particular, we do not assume that the additive noise generated during transmission 
is a Gaussian random variable.
So, we need a non-parametric estimation, which has been resolved by Kolmogorov-Smirnov test \cite{Kolmogorov,Smirnov}.
Combining a suitable exponential upper bound for the leaked information
and the above error evaluation,
we derived finite-length security analysis as Theorem \ref{11-17T}.
As Fig \ref{leaked1}, we give a numerical calculation for the upper bound given in 
\eqref{28-2X} in  a typical example.


When Eve breaks the quasi static condition,
she can change the artificial noise dependently of the pulse.
In this case, if Alice and Bob pre-agree which pulses are used for samples,
Eve can insert the large artificial noise only to the non-sampling pulses
so that the condition \eqref{4-24-3} does not hold in the non-sampling pulses 
without detection by Alice and Bob.
Then, Eve can succeed in eavesdropping without detection by Alice and Bob. 
Currently, we might not have such a technology, however, we cannot deny such an eavesdropping in future.
Fortunately, in our protocol, Alice and Bob do not fix the sampling pulse priorly,
they choose the sample pulse after the transmission from Alice to Bob as the random sampling, whose security guaranteed by authentication.
Then, Eve cannot selectively insert the artificial noise.
The same effect is utilized in BB84 protocol of quantum key distribution (QKD) \cite{BB84}.
However, the errors in Bob's observations is not necessarily subject to an identical and independent distribution.

In the case of QKD, in the early stage of their analysis, 
they assume that the errors is subject to an identical and independent distribution.
The attack under this condition is called the collective attack in the QKD \cite{BBBGM}.
Later, they removed this assumption by using hypergeometric distribution \cite{S-P,H06}
because the behavior of this random sampling in the discrete variables case can be discussed by 
hypergeometric distribution.
However, since our system employs the continuous variable, this type evaluation is not so easy.
Therefore, removing this assumption in our case
is beyond the focus of this paper and is an interesting future problem.

Section \ref{S5E} made numerical calculations only when 
$Y$ is subject to the Gaussian distribution.
In a practice, there is a possibility that $Y$ does not obey the Gaussian distribution.
Hence, it is another interesting future study to make numerical calculations 
for another type of distribution for $Y$ like Section \ref{S5E}.

Further, there still exits a possibility that Eve can break the assumptions of our model.
Such a possibility might be realized in the following two cases. 
(i) Eve can concentrate her resource to break the assumptions.
If Eve prepares a very expensive measurement device or too many expensive measurement devices, the assumptions are broken.
(ii) Eve luckily gets a very large value of $ a_E$ due to the interference effect.
To resolve this problem, the forthcoming paper \cite{HOKC} proposes to combine secure network coding \cite{Cai2002,Cai06a,Cai06,HLKMEK,JLHE}
and secure wireless communication.
That is, we consider secure network coding whose communications on the edges are realized by our secure wireless communication.
In this case, for an eavesdropping, Eve has to break the assumptions of our model in multiple wireless communication channels.
For the case (i), Eve has to distribute her devices in these communication channels.
This combination increases the difficulty of the eavesdropping.
For the case (ii), Eve needs to be lucky in multiple wireless communication channels.
Usually, the event of a large value $a_E$ might be regarded to be independent of the same event with the different point.
Hence, the above possibility becomes very small.
Therefore, we can decrease the possibility that Eve makes eavesdropping
by combining secure network coding with our result.

\section*{Acknowledgments}
The author is very grateful to Professor Hideichi Sasaoka and Professor Hisato Iwai for helpful discussions and informing the references \cite{SC,Trappe,Zeng,WX}.
He is also grateful to
Professor \'{A}ngeles Vazquez-Castro,
Professor Matthieu Bloch,
Professor Shun Watanabe, Professor Himanshu Tyagi, and Dr. Toyohiro Tsurumaru for helpful discussions and comments.
The works reported here were supported in part by 
the JSPS Grant-in-Aid for Scientific Research (B) No. 16KT0017 and  (A) No.17H01280,
the Okawa Research Grant
and
Kayamori Foundation of Informational Science Advancement.


\bibliographystyle{IEEE}

\end{document}